\documentclass[runningheads,anonymous]{llncs}
\usepackage{graphicx}
\usepackage{framed}
\usepackage{syntax}
\usepackage{amssymb}
\usepackage{color}
\usepackage{amsmath}
\usepackage{stmaryrd}
\usepackage{makecell}
\usepackage{comment}
\usepackage[noline,boxed,ruled,noend,linesnumbered]{algorithm2e}
\usepackage{wrapfig}
\usepackage{appendix}
\usepackage{enumerate,paralist}
\usepackage[hidelinks]{hyperref}
\usepackage{multirow}
\usepackage[T1]{fontenc}

\newcommand{\kwtick}{\mathtt{pre}}
\newcommand{\kwinvoke}{\mathtt{invoke}}
\newcommand{\kwsample}{\mathtt{sample}}
\newcommand{\kwin}{\mathtt{in}}
\newcommand{\kwwith}{\bigoplus}
\newcommand{\com}{\mathsf{comm}}
\newcommand{\expr}{\mathsf{size}}
\newcommand{\pp}{\mathsf{expr}}
\newcommand{\dist}{\mathsf{dist}}
\newcommand{\etf}{\mathsf{Tf}}

\newcommand{\emp}{\mathsf{emp}}
\newcommand{\body}{\mathsf{body}}
\newcommand{\rbody}{\mathsf{call}}

\newcommand{\pspace}{\mathrm{\Omega}}
\newcommand{\probspace}{(\pspace,\mathcal F,\Pr)}
\newcommand{\supp}{\mathrm{supp}}
\newcommand{\ex}{\mathsf{Ex}}
\newcommand{\prrlang}{{\sl LRec}}

\newcommand{\coef}{c}
\newcommand{\spolya}{p}
\newcommand{\sloga}{q}
\newcommand{\spolyn}{u}
\newcommand{\slogn}{v}
\newcommand{\coeff}{\coef_f}
\newcommand{\coeft}{\coef_t}
\newcommand{\logfa}{\sloga_{f}}
\newcommand{\logfn}{\slogn_{f}}
\newcommand{\polyfa}{\spolya_{f}}
\newcommand{\polyfn}{\spolyn_{f}}
\newcommand{\logta}{\sloga_{t}}
\newcommand{\logtn}{\slogn_{t}}
\newcommand{\polyta}{\spolya_{t}}
\newcommand{\polytn}{\spolyn_{t}}

\newcommand{\cont}{\mathbf{K}}
\newcommand{\cost}{C}
\newcommand{\conf}{\mu}

\newcommand{\with}{\kwwith_{i=1}^k c_i\!:\!\com_i}
\newcommand{\sample}{\kwsample\ v\leftarrow\dist\ \kwin\ \{\body\}}

\newcommand{\transprob}{\mathbf{P}}

\renewcommand{\paragraph}[1]{\smallskip\noindent{\textbf{\textit{#1.}}}}

\newcommand{\Nset}{\mathbb N}

\newcommand{\condexpv}[2]{\expv\left[#1\ |\ #2\right]}
\newcommand{\dtsp}{\Gamma}

\newcommand{\coststack}{\mbox{\sl S}}

\newcommand{\probm}{\Pr}

\newcommand{\epxv}{\mathbb E}
\newcommand{\expv}{\epxv}

\renewcommand{\color}[1]{}

\newenvironment{talign}
 {\let\displaystyle\textstyle\align}
 {\endalign}
\newenvironment{talign*}
 {\let\displaystyle\textstyle\csname align*\endcsname}
 {\endalign}

\begin{document}

\title{Automated Tail Bound Analysis for Probabilistic Recurrence Relations\thanks{Mainland Chinese authors are ordered by contribution, while Austrian and Chinese Hong Kong authors
are ordered alphabetically. {The code and benchmarks are available at https://github.com/boyvolcano/PRR}}\protect}

\author{Yican Sun\inst{1} \and
Hongfei Fu \inst{2} \thanks{Corresponding Author} \and
Krishnendu Chatterjee \inst{3}\and\\
Amir Kafshdar Goharshady \inst{4}
}
\institute{School of Computer Science, Peking University, China\\ \email{sycpku@pku.edu.cn}\and
Department of Computer Science, Shanghai Jiao Tong University, China\\ \email{fuhf@cs.sjtu.edu.cn}\and
IST Austria\\ \email{krishnendu.chatterjee@ist.ac.at}\and
The Hong Kong University of Science and Technology\\ \email{goharshady@cse.ust.hk}
}

\maketitle

\begin{abstract}
Probabilistic recurrence relations (PRRs) are a standard formalism for describing the runtime of a randomized algorithm. Given a PRR
and a time limit $\kappa$, 
we consider the classical concept of tail probability $\probm[T \ge \kappa]$, i.e., the probability that the randomized runtime $T$ of the PRR exceeds the time limit $\kappa$. Our focus is the formal analysis of tail bounds that aims at finding a tight asymptotic upper bound $u \geq \probm[T\ge\kappa]$ in the time limit $\kappa$.
To address this problem, the classical and most well-known approach is the cookbook method by Karp (JACM 1994), while other approaches are mostly limited to deriving tail bounds of specific PRRs via involved custom analysis.

In this work, we propose a novel approach for deriving exponentially-decreasing tail bounds (a common type of tail bounds) for PRRs whose preprocessing time and random passed sizes observe discrete or (piecewise) uniform distribution and whose recursive call is either a single procedure call or a divide-and-conquer. We first establish a theoretical approach via Markov's inequality, and then instantiate the theoretical approach with a template-based algorithmic approach via a refined treatment of exponentiation. 
Experimental evaluation shows that our algorithmic approach is capable of deriving tail bounds that are (i) asymptotically tighter than Karp's method, (ii) match the best-known manually-derived asymptotic tail bound for QuickSelect, and (iii) is only slightly worse (with a $\log\log n$ factor) than the manually-proven optimal asymptotic tail bound for QuickSort.
Moreover, our algorithmic approach handles all examples
(including realistic PRRs such as QuickSort, QuickSelect, DiameterComputation, etc.) in less than 0.1 seconds, showing that our approach is efficient in practice.
\end{abstract}

\section{Introduction}
\label{sec:intro}

Probabilistic program verification is a fundamental area in formal verification \cite{BaierBook}. It extends the classical (non-probabilistic) program verification by considering randomized computation in a program and hence can be applied to the formal analysis of probabilistic computations such as probabilistic models~\cite{goodman2008church}, randomized algorithms~\cite{DBLP:conf/itp/Tassarotti018,expectedRecurrence,poplAbstraction,poplSensitivity}, etc. In this line of research, verifying the time complexity of probabilistic recurrence relations (PRRs) is an important subject~\cite{DBLP:conf/itp/Tassarotti018,expectedRecurrence}. PRRs are a simplified form of recursive probabilistic programs and extend recurrence relations by incorporating randomization such as randomized preprocessing and divide-and-conquer. They are widely used in analyzing the time complexity of randomized algorithms (e.g., QuickSort~\cite{DBLP:journals/cacm/Hoare61b}, QuickSelect~\cite{DBLP:journals/cacm/Hoare61a}, and DiameterComputation~\cite[Chapter 9]{DBLP:books/cu/MotwaniR95}). Compared with probabilistic programs, PRRs abstract away detailed computational aspects, such as problem-specific divide-and-conquer and data-structure manipulations, and include only key information on the runtime of the underlying randomized algorithm.
Hence, PRRs provide a clean model for time-complexity analysis of randomized algorithms and randomized computations in a general sense.

 In this work, we focus on the formal analysis of PRRs and consider the fundamental problem of tail bound analysis that aims at bounding the probability that a given PRR does not terminate within a prescribed time limit. In the literature, prominent works on tail bound analysis include the following. First, Karp proposed a classic ``cookbook'' formula~\cite{DBLP:journals/jacm/Karp94} similar to Master Theorem. This method is further improved, extended, and mechanized by follow-up works~\cite{DBLP:journals/tcs/ChaudhuriD97,DBLP:conf/itp/Tassarotti018,DBLP:series/txtcs/BertotC04}. While Karp's method has a clean form and is easy to use and automate, the bounds from the method are known to be not tight (see e.g.~ \cite{DBLP:journals/jal/McDiarmidH96,QuickSelectSOTA}). Second, the works \cite{DBLP:journals/jal/McDiarmidH96} and resp. \cite{QuickSelectSOTA} performed ad-hoc custom analysis to derive asymptotically tight tail bounds for the PRRs of QuickSort and resp. QuickSelect, respectively. These methods require manual effort and do not have the generality to handle a wide class of PRRs.

From the literature, an algorithmic approach capable of deriving tight tail bounds over a wide class of PRRs is a major unresolved problem. Motivated by this challenge,  we have the following contributions to this work:

\begin{compactitem}
\item Based on Markov's inequality, we propose a novel theoretical approach to derive exponentially-decreasing tail bounds, a common type for many randomized algorithms. {\color{red} We further show that our theoretical approach can always derive an exponentially-decreasing tail bound at least as tight as Karp's method under mild assumptions.}

\item From our theoretical approach, we propose a template-based algorithmic approach for a wide class of PRRs that have (i) common probability distributions such as (piecewise) uniform distribution and discrete probability distributions and  (ii) either a single call or a divide-and-conquer for the form of the recursive call.
The technical novelties in our algorithm lie in a refined treatment of the estimation of the exponential term arising from our theoretical approach via integrals, suitable over-approximation, and the monotonicity of the template function.

\item Experiments show that our algorithmic approach derives asymptotically tighter tail bounds when compared with Karp's method. Furthermore, the tail bounds derived from our approach match the best-known bound for QuickSelect~\cite{QuickSelectSOTA}, and are only slightly worse by a $\log\log n$ factor against the optimal manually-derived bound for QuickSort~\cite{DBLP:journals/jal/McDiarmidH96}.
Moreover, our algorithm synthesizes each of these tail bounds in less than 0.1 seconds and is efficient in practice.

\end{compactitem}

A limitation of our approach is that we do not consider the transformation from a realistic implementation of a randomized algorithm into its PRR representation. However, such a transformation would require examining a diversified number of randomization patterns (e.g., randomized divide-and-conquer) in randomized algorithms and thus is an orthogonal direction. In this work, we focus on the tail bound analysis and present a novel approach to address this problem.

\vspace{-1ex}
\section{Preliminaries}
\label{sec:sprr}

Below we present necessary background in probability theory and the tail bound analysis problem we consider.

A \emph{probability space} is a triple $(\mathrm{\pspace}, \mathcal F, \mathrm{Pr})$ such that $\pspace$ is a non-empty set termed as the \emph{sample space}, $\mathcal F$ is a \emph{$\sigma$-algebra} over $\pspace$ (i.e., a collection of subsets of $\pspace$ that contains the empty set $\emptyset$ and is closed under complement and countable union), and $\mathrm{Pr}(\cdot)$ is a \emph{probability measure} on $\mathcal F$ (i.e., a function $\mathcal F\to [0,1]$ such that  $\Pr(\pspace)=1$ and for every pairwise disjoint set-sequence $A_1,A_2,\dots$ in $\mathcal{F}$, we have that $\sum_{i\ge 1}\Pr(A_i)=\Pr\left(\bigcup_{i\ge 1}A_i\right)$.

A \emph{random variable} $X$ from a probability space $(\pspace, \mathcal F, \mathrm{Pr})$ is an $\mathcal F$-measurable function $X:\pspace\to \mathbb R$, i.e., for every $d\in \mathbb R$, we have that $\{\omega\in \pspace\mid X(\omega)<d\}\in \mathcal F$. We denote $\mathbb E[X]$ as its expected value; formally, we have $\mathbb{E}[X]:=\int X \,\mathrm{d}\mathrm{Pr}$.
A \emph{discrete probability distribution}~(DPD) over a countable set $U$ is a function $\eta: U\to [0,1]$, such that $\sum_{u\in U}{\eta(u)}=1$. The \emph{support} of the DPD is defined as $\supp(\eta):=\{u\in U\mid \eta(u)>0\}$. We abbreviate finite-support DPD as FSDPD.

A \emph{filtration} of probability space $\probspace$ is an infinite sequence of $\{\mathcal F_n\}_{n\ge 0}$ of $\sigma$-algebra over $\pspace$ such that $\mathcal F_{n}\subseteq \mathcal F_{n+1}\subseteq \mathcal F$ for every $n\ge 0$. Intuitively, it models the information at the $n$-th step.
A \emph{discrete-time stochastic process} is an infinite sequence $\mathrm{\Gamma} = \{X_n\}_{n\ge 0}$ of random variables from the probability space $\probspace$. The process $\mathrm{\Gamma}$ is \emph{adapted} to a filtration $\{\mathcal F_{n}\}_{n\ge 0}$ if for all $n\ge 0$, $X_n$ is $\mathcal F_n$-measurable. 
Given a filtration $\{\mathcal F_{n}\}_{n\ge 0}$, a \emph{stopping time}  is a random variable $\tau: \Omega\to \mathbb N$, such that for every $n \ge 0$, $\{\omega\in \pspace\mid \tau(\omega)\le n\}\in \mathcal F_n$.

A discrete-time stochastic process $\Gamma=\{X_n\}_{n\in\Nset}$ adapted to a filtration $\{\mathcal{F}_n\}_{n\in\Nset}$ is a \emph{martingale} (resp. \emph{supermartingale})
if for every $n\in\Nset$, $\expv[|X_n|]<\infty$ and it holds a.s. that $\condexpv{X_{n+1}}{\mathcal{F}_n} = X_n$ (resp. $\condexpv{X_{n+1}}{\mathcal{F}_n}\le X_n$).
Intuitively, a martingale (resp. supermartingale) is a discrete-time stochastic process in which for an observer who has seen the values of $X_0, \ldots, X_n$, the expected value at the next step, i.e.~$\condexpv{X_{n+1}}{\mathcal{F}_n}$, is equal to (resp. no more than) the last observed value $X_n$. Also, note that in a martingale, the observed values for $X_0, \ldots, X_{n-1}$ do not matter given that $\condexpv{X_{n+1}}{\mathcal{F}_n} = X_n.$ In contrast, in a supermartingale, the only requirement is that $\condexpv{X_{n+1}}{\mathcal{F}_n} \leq X_n$ and hence $\condexpv{X_{n+1}}{\mathcal{F}_n}$ may depend on $X_0, \ldots, X_{n-1}.$ Also, note that $\mathcal{F}_n$ might contain more information than just the observations of $X_i$'s.
\begin{example}
	Consider the classical gambler's ruin: a gambler starts with $Y_0$ dollars of money and bets continuously until he loses all of his money. If the bets are unfair, i.e.~the expected value of his money after a bet is less than its expected value before the bet, then the sequence $\{Y_n\}_{n \in \mathbb{N}_0}$ is a supermartingale. In this case, $Y_n$ is the gambler's total money after $n$ bets. On the other hand, if the bets are fair, then $\{Y_n\}_{n \in \mathbb{N}_0}$ is a martingale. See Appendix~\ref{app:martingales} for another illustrative example of martingales. \qed
\end{example}
We refer to standard textbooks (such as~\cite{williams1991probability,Billingsley:book}) for a detailed treatment of all the concepts illustrated above.

\vspace{-1ex}
\subsection{Probabilistic Recurrence Relations}

In this work, we focus on probabilistic recurrence relations (PRRs) that describe the runtime behaviour of a single recursive procedure.
Instead of having a direct syntax for a PRR, we propose a
mini programming language \prrlang{} that captures a wide class of PRRs 
{\color{red} that have common probability distributions such as (piecewise) uniform distributions and discrete probability distributions, and whose recursive call consists of
either a procedure call or two procedure calls in a divide-and-conquer style. We present the grammar of \prrlang{} in  Figure~\ref{fig:grammar}.}

\begin{figure}[htbp]
    \centering
    
    \begin{footnotesize}
    \begin{tabular}{rrrl}
\textbf{(PRR)} & \textsf{proc} &::=& $\mathtt{def}\  p(n ; c_p)=\{\mathsf{comm}\}$\\
\textbf{(Command)} & \textsf{comm} &::=& $\sample\mid \with$\\
\textbf{(Recursive Body)} & $\body$ &::=&  $\mathtt{pre}(\pp);\kwinvoke\ \rbody$\\
\textbf{(Recurive Call)} & $\rbody$ &::=&  $p(v);p(\expr-v)\,\mid\,p(v)\,\mid\,p(\expr-v)$
\\
&&& (where $\expr$ is either $\lfloor \frac nb\rfloor+c$ or $\lceil \frac nb\rceil+c$)\\\textbf{(Distribution)} & $\dist$ &::=& \texttt{uniform}$(n)$ $\mid$ \texttt{muniform}$(n)$ $\mid$ \texttt{discrete} $\mid$ $\ldots$\\
        \textbf{(Expression)} & $\pp$ &::=& ${v}\mid {v}^{-1} \mid \ln {v}\mid n\mid \ln n\mid n^{-1}\mid {c}$\\& & &  $\mid \pp\!+\!\pp\mid \pp\!-\!\pp\mid \pp\!\times\!\pp$
    \end{tabular}
    \end{footnotesize}
    \caption{The Grammar of \prrlang}
    \vspace{-1.5em}
    \label{fig:grammar}
\end{figure}

In the grammar, we have two positive-integer valued variables $n,v$ which stand for the input size and the sampled value in the randomization of the passed size to the recursive calls of a procedure, respectively. {\color{red} We use $b>0,c,c_p$ to denote integer constants, and use $p$ to denote the name of the single procedure in the PRR.}
We consider arithmetic expressions $\pp$ as polynomials over $v,v^{-1},\ln v$ and $n,n^{-1},\ln n$ (which we call \emph{pseudo-polynomials} in this work) and common probability distributions, including (i) the uniform distribution \texttt{uniform}$(n)$ over $\{0,1,\dots, n-1\}$, (ii) the piecewise uniform distribution \texttt{muniform}$(n)$ that returns $\max\{i,n-i-1\}$ where $i$ observes the uniform distribution $\mathtt{uniform}(n)$, and (iii) any FSDPD (indicated by \textsf{discrete}) whose probabilities and values are constants and pseudo-polynomials, respectively. {\color{red} We also support other piecewise uniform distribution, e.g, the distribution that each $v\in \{0,\ldots,n/2\}$ has probability $\frac{2}{3n}$ and each $v\in \{n/2+1,\ldots,n-1\}$ has probability $\frac{4}{3n}$.}

The nonterminal \textsf{proc} generates the PRR in the form $\mathtt{def}\  p(n ; c_p)=\{\mathsf{comm}\}$, for which
$c_p$ is an integer constant as the threshold of recursion, meaning that the procedure halts immediately when $n< c_p$, and $\mathsf{comm}$ is the function body of the procedure.
The nonterminal \textsf{comm} generates all statements with one of the two forms as follows.
\begin{compactitem}
    \item A sampling statement (indicated by \textsf{sample}) followed by first a special expression $\mathtt{pre}(\pp)$ that stands for the preprocessing time of $\pp$ amount, then the recursive calls generated by the nonterminal $\rbody$. 
    \item A probabilistic choice in the form $\with$ where each statement $\mathsf{comm}_i$ is executed with probability $c_i$.
\end{compactitem}

We restrict the recursive calls to be either a single recursive call $p(v)$ or $p(\expr-v)$, or a divide-and-conquer composed of two consecutive recursive calls $p(v)$ and $p(\expr-v)$, for which we consider a general setting that the relevant overall size $\expr$ is in the form of the input size $n$ divided by some positive integer $b$ with possibly an offset $c$. Choosing $b=1,c=-1$ means the normal situation that the overall size is $n-1$, i.e., removing one element from the original input.

 Given a PRR $p$, we use $\mathsf{func}(p)$ to represent its function body.
 
 We always assume that the given PRR is \emph{well-formed}, i.e., every $c_i$ in a probabilistic choice is within $[0,1]$ and every random passed size (e.g. $v,\mathsf{size}-v$) falls in $[0,n]$.
 Below, we present two examples for PRRs.

\begin{example}[{QuickSelect}]
	\label{ex:quickselect}
	Consider the problem of finding the $d$-th smallest element in an unordered array of $n$ distinct elements. A classical randomized algorithm for solving this problem is {QuickSelect}~\cite{DBLP:journals/cacm/Hoare61a} with $O(n)$ expected running time. We present the detail of this algorithm in Appendix~\ref{app:examples}. We model the algorithm as the following PRR: 
\begin{small}
$$\textstyle
\mathtt{def}\ p(n;2)= \{\mathtt{sample}\ v\leftarrow\mathtt{muniform}(n)\ \mathtt{in}\ \{\kwtick(n);\ \mathtt{invoke}\ p(v);\}\}
$$
\end{small}
Here, we use $p(n;2)$ to represent the number of comparisons performed by {QuickSelect} over an input of size $n$, and $v$ is the variable that captures the size of the remaining array that has to be searched recursively. It observes as the value $\max\{i, n\!-1\!-i\}$  where the value of $i$ is sampled uniformly from $\{0, \ldots, n\!-\!1\}$, we use $\mathtt{muniform}(n)$ to represent this distribution. \qed
\end{example}

\begin{example}[{QuickSort}]
	\label{ex:quicksort}
	Consider the classical problem of sorting an array of $n$ distinct elements.
A well-known randomized algorithm for solving this problem is {QuickSort}~\cite{DBLP:journals/cacm/Hoare61b}.
We model the algorithm as the following PRR. The detail of this algorithm is also presented in Appendix~\ref{app:examples}. 
\begin{small}
$$
\mathtt{def}\ p(n;2)= \{\mathtt{sample}\ v\leftarrow\mathtt{uniform}(n)\ \mathtt{in}\ \{\kwtick(n);\ \mathtt{invoke}\ p(v);p(n-1-v);\}\}
$$
\end{small}
Here, $v$ and $n-1-v$ capture the sizes of the two sub-arrays.
\qed
\end{example}

Below we present the semantics of a PRR in a nutshell.
We relegate the details of the semantics in Appendix~\ref{app:semantics}.
Consider a PRR generated by \prrlang{} with the procedure name $p$, a \emph{configuration} $\sigma$ is a pair $\sigma=(\mbox{\sl comm}, \widehat{n})$ where $\mbox{\sl comm}$ represents the current statement to be executed and $\widehat{n}\ge c_p$  is the current value for the variable $n$.
A \emph{PRR state}  $\conf$ is a triple $\langle \sigma,\cost,\cont \rangle$ for which:
\begin{compactitem}
\item $\sigma$ is either a configuration, or $\mathsf{halt}$ for the termination of the whole PRR.
\item $\cost\ge 0$ records the cumulative preprocessing time so far.
\item $\cont$ is a stack of configurations that remain to be executed.
\end{compactitem}
 We use $\mathsf{emp}$ to denote an empty stack, and say that a PRR state $\langle \sigma,C,\cont \rangle$ is \emph{final} if $\cont=\emp$ and $\sigma=\mathsf{halt}$. Note that in a final PRR state $\langle \mathsf{halt}, C,\emp \rangle$, the value $C$ represents the total execution runtime  of the PRR.
The semantics of the PRR is defined as a discrete-time Markov chain whose state space is the set of all PRR states and whose transition function $\transprob$, where $\transprob(\conf, \conf')$ is the probability that the next PRR state is $\conf'$ given the current PRR state is $\conf=((\mbox{\sl comm}, \widehat{n}),C,\cont)$.
The probability is determined by the following cases.

\begin{itemize}
\item For final PRR states $\conf$, $\transprob(\conf,\conf):=1$ and $\transprob(\conf,\conf'):=0$ for other $\conf'\ne \conf$. This means that the PRR stays at termination once it terminates.
\item In the divide-and-conquer case $\mbox{\sl comm}=\kwsample\ v\leftarrow \mbox{\sl dist}\ \kwin\ \{\mathtt{pre}(e);\\\kwinvoke\ p(v);p(s-v)\}$, we first sample $v$ from the distribution $\mbox{\sl dist}$. Then, with probability $\mbox{\sl dist}(v)$, we accumulate the preprocessing time $e$ into the cumulative processing time $C$. We recursively invoke $p(v)$ and push the remaining task $p(s-v)$ into the stack. The probability for the single recursion case is defined analogously. The only difference is that there is no need to push some recursive call into the stack in the single recursion case.
\item  In the case $\mbox{\sl comm}=\kwwith_{i=1}^kc_i\!:\!\mbox{\sl comm}_i$,
we have that $\transprob(\conf, \conf_i) = {c_i}$ for each $1\le i\le k$ for which we have
$\conf_i:=((\mbox{\sl comm}_i, \widehat{n}),C,\cont)$.
\end{itemize}
With an initial PRR state
$((\mathsf{func}(p), n^*), 0, \emp)$
where $n^*\ge c_p$ is the input size, the Markov chain induces a probability space
where the sample space is the set of all infinite sequences of PRR states, the $\sigma$-algebra is generated by all \emph{cylinder sets} over infinite sequences of PRR states
, and the probability measure is uniquely determined by the  transition function $\transprob$. We refer to \cite{BaierBook} for details. We use $\Pr_{n^*}$ for the probability measure where $n^*\ge c_p$ is the input size.

We further define the random variable $\tau$ such that
for any infinite sequence of PRR states $\rho = \conf_0, \conf_1, \dots,\conf_t,\dots$ {\color{green} with each $\conf_t=((\mbox{\sl comm}_t, \widehat{n}_t),C_t,\cont_t)$}, {\color{green} $\tau(\rho)$} equals the first moment that the sequence reaches a final PRR state, i.e., ${\color{green}\tau(\rho)} = \inf\{t\mid \mbox{the PRR state }\conf_t\mbox{ is final}\}$, for which $\inf\emptyset = \infty$.
We will always ensure that $\tau$ is almost-surely finite, i.e., $\Pr_{n^*}(\tau<\infty)=1$). Note that the {\color{green}random} cumulative processing time $C_\tau$ in the PRR state $\mu_{\tau}\in \rho$ is the total execution time of the given PRR.

We formulate the tail bound analysis over PRRs as follows.
Given a time limit $\alpha\cdot \kappa(n^*)$ symbolic in the initial input $n^*$ and the coefficient $\alpha$, 
the goal of tail bound analysis is to infer an upper bound $u(\alpha,n^*)$ symbolic in $n^*$ and $\alpha$ such that for every input size $n^*$ and plausible value for $\alpha$, we have that
\begin{equation}
\label{eq:problem}
\textstyle\Pr_{n^*}[C_{\tau}\ge \alpha\cdot \kappa(n^*)]\le u(\alpha,n^*).
\end{equation}

As tails bounds are often evaluated asymptotically, we focus on deriving tight $u(\alpha,n^*)$
when $\alpha,n^*$ are sufficiently large. To compare the magnitude of two tail bounds, 
we follow the straightforward way that first treats $\alpha$ as a fixed constant and compares the bounds over $n^*$, and then if the magnitude over $n^*$ is identical,
we take a further comparison over the magnitude on the coefficient $\alpha$. 

\begin{example}[Our result on QuickSelect]\label{ex:quickselelct-overall}
    Continue with Example~\ref{ex:quickselect}, suppose the user is interested in the tail bound~$\Pr[C_{\tau}\ge \alpha\cdot n^*]$, where $C_{\tau}$ is the running time of the QuickSelect algorithm over an array with length $n^*$. Then, Karp's method produces the symbolic tail bound as follows. $$\Pr[C_{\tau}\ge \alpha\cdot n^*]\le \exp(1.15-0.28\cdot \alpha)$$
    However, our method can produce the following tail bound. $$\Pr[C_{\tau}\ge \alpha\cdot n^*]\le \exp(2\cdot \alpha-\alpha\cdot \ln \alpha)$$
    Note that our method produces tail bounds with a better magnitude on $\alpha$. \qed
\end{example}

\begin{example}[Our result on QuickSort]\label{ex:quickselelct-overall}
    Continue with Example~\ref{ex:quicksort}, consider the tail bound~$\Pr[C_{\tau}\ge \alpha\cdot n^*\cdot \ln n^*]$, where $C_{\tau}$ is the running time of QuickSort over a length-$n^*$ array. Then, Karp's method produces the symbolic tail bound as: $$\Pr[C_{\tau}\ge \alpha\cdot n^*\cdot \ln n^*]\le \exp(0.5-0.5\cdot \alpha),$$
    while our method can produce the bound as: $$\Pr[C_{\tau}\ge \alpha\cdot n^*\cdot \ln n^*]\le \exp((4-\alpha)\cdot \ln n^*)$$
    Note that our method produces tail bounds with a better magnitude on $n^*$.\qed
\end{example}

\vspace{-1ex}
\section{Exponential Tail Bounds via Markov's Inequality}
\label{sec:theory}

In this section, we demonstrate our theoretical approach for deriving exponentially decreasing tail bounds 
based on Markov's inequality. 
Due to the space limitation,
We relegate all proof details into Appendix~\ref{appendix:theory}.

Before illustrating our approach, we first translate a PRR in the language \prrlang{} with the single procedure $p$ into the canonical form as follows.
\begin{equation}
\label{eq:canonical-prr}
p(n;c_p)=\kwtick(\coststack(n));  \kwinvoke\  p(\expr_1(n));\ldots; p(\expr_r(n))
\end{equation}
where (i) $\coststack(n)$ is a random variable related to the input size $n$ that represents the randomized pre-processing time and observes a probability distribution resulting from a discrete probability choice of piecewise uniform distributions, and (ii) $\kwinvoke\  p(\expr_1(n));\ldots; p(\expr_r(n))$ is a statement that is either a single recursive call $p(\expr_1(n))$ or a divide-and-conquer $p(\expr_1(n)); p(\expr_2(n))$ upon the resolution of the randomization. For the latter, we use a random variable $r$ (which is either $1$ or $2$) to represent the number of recursive calls.

The translation can be implemented by a straightforward recursive procedure $\etf(n,{\sl Prog})$ that takes on input a positive integer $n$ (as the input size) and a statement ${\sl Prog}$ (generated by the nonterminal $\textsf{comm}$) to be processed, see Appendix~\ref{appendix:theory} for details.
Note that the procedure $\etf(n,{\sl Prog})$ outputs the \emph{joint} distribution of the random value $\coststack(n)$ and the random recursive call $p(\expr_1(n));\ldots;\\ p(\expr_r(n))$
i.e., they
may not be independent.

Our theoretical approach then works directly on the canonical form (\ref{eq:canonical-prr}). It consists of two major steps to derive an exponentially-decreasing tail bound.
In the first step, we apply Markov's inequality and reduce the tail bound analysis problem to the over-approximation of the moment generating function $\mathbb E[\exp(t\cdot C_{\tau})]$ where $C_{\tau}$ is the cumulative pre-processing time defined previously and $t>0$ is a scaling factor that aids the derivation of the tail bound.
In the second step, we apply Optional Stopping Theorem (a classical theorem in martingale theory) to over-approximate the expected value $\mathbb E[\exp(t\cdot C_{\tau})]$. Below
we fix an PRR with procedure $p$ in the canonical form (\ref{eq:canonical-prr}), and a time limit  $\alpha
\cdot\kappa(n^*)$.  

Our first step applies Markov's inequality. Our approach relies on the well-known exponential form of Markov's inequality below. 

\begin{theorem}
\label{thm:markov}
    For every random variable $X$ and any scaling factor $t>0$, we have that $\mathbb \Pr[X\ge d]\le \expv[\exp(t\cdot X)]/\exp(t\cdot d)$.
\end{theorem}

The detailed application of Markov's inequality to tail bound analysis requires to choose a scaling factor $t:=t(\alpha,n)$ symbolic in $\alpha$ and $n$. After choosing the scaling factor, Markov's inequality gives the following tail bound:
\begin{equation}\label{ineq:markov}
\Pr[C_{\tau}\ge \alpha\cdot \kappa(n^*)]\le \mathbb E[\exp(t(\alpha,n^*)\cdot C_{\tau})]/\exp(t(\alpha,n^*)\cdot \alpha\cdot \kappa(n^*)).
\end{equation}

The role of the scaling factor $t(\alpha,n^*)$ is to scale the exponent in the term $\exp(\kappa(\alpha,n^*))$, and this is in many cases necessary as a tail bound may not be exponentially decreasing directly in the time limit $\alpha\cdot \kappa(n^*)$.

An unsolved part in the tail bound above is the estimation of the expected value $\mathbb E[\exp(t(\alpha,n^*)\cdot C_{\tau})]$. Our second step over-approximates the expected value $\mathbb E[\exp(t(\alpha,n^*)\cdot C_{\tau})]$. To achieve this goal, we impose a constraint on the scaling factor $t(\alpha, n)$ and an extra function $f(\alpha, n)$
and show that once the constraint is fulfilled, then one can derive an upper bound for $\mathbb E[\exp(t(\alpha,n^*)\cdot C_{\tau})]$ from $t(\alpha, n)$ and $f(\alpha, n)$. The theorem is proved via Optional Stopping Theorem. The theorem requires the almost-sure termination of the given PRR, a natural prerequisite of exponential tail bound. In this work, we consider PRRs with finite termination time that implies the almost-sure termination.

\begin{theorem}\label{thm:martingale}
Suppose we have
functions $t, f: [0,\infty)\times \mathbb{N}\rightarrow [0,\infty)$
such that
    \begin{talign}
    \label{eq:constraint}
   &\mathbb E[\exp(t(\alpha,n)\cdot \ex(n\mid f))]\le \exp(t(\alpha,n)\cdot f(\alpha,n))
    \end{talign}
for all sufficiently large $\alpha,n^*>0$ and all $c_p\le n\le n^*$, where
$$
\textstyle \ex(n\mid f):=\coststack(n)+\sum_{i=1}^{r}f(\alpha,\expr_i(n)).
$$
Then for $t_*(\alpha,n^*):=\min_{c_p\le n\le n^*} t(\alpha,n)$, we have that $$\mathbb E[\exp(t_*(\alpha,n^*)\cdot C_{\tau})]\le \mathbb E[\exp(t_*(\alpha,n^*)\cdot f(\alpha,n^*))].$$
Thus, we obtain the upper bound $u(\alpha,n^*):=\exp(t_*(\alpha,n^*)\cdot (f(\alpha,n^*)-\alpha\cdot \kappa(n^*)))$ for the tail bound in (\ref{eq:problem}).
\end{theorem}

\noindent\emph{Proof sketch.}
We fix a procedure $p$, and some sufficiently large $\alpha$ and $n^*$.
In general, we apply the martingale theory to prove this theorem. To construct a martingale, we need to make two preparations.

First, by the convexity of $\exp(\cdot)$, substituting $t(\alpha,n)$ with $t_*(\alpha,n^*)$ in (\ref{eq:constraint}) does not affect the validity of (\ref{eq:constraint}).

Second, given an infinite sequence of the PRR states $\rho=\mu_0,\mu_1,\dots$ in the sample space, we consider the subsequence $\rho'=\mu_0',\mu_1',\dots$ as follows, where we represent $\mu_i'$ as $(({\sf func}(p),\hat{n}_i'),C_i',\mathbf{K}_i')$.
It only contains states that are either final or at the entry of $p$, i.e., ${\sl comm}={\sf func}(p)$.
We define $\tau':=\inf\{t: \mu'_t\text{ is final}\}$, then it is straightforward that $C'_{\tau'}=C_{\tau}$.
We observe that $\mu_{i+1}'$ represents the recursive calls of $\mu_i'$. Thus, we can  characterize the conditional distribution $\mu_{i+1}'\mid \mu_i$ by the transformation function $\etf(\hat{n},\mathsf{func}(p))$ as follows. 
\begin{compactitem}
\item We first draw $(\coststack, \expr_1,\expr_2,r)$ from $\etf(\hat{n}_i',\mathsf{func}(p))$.
\item We accumulate $\coststack$ into the global cost. If there is a single recursion ($r=1$), we invoke this sub-procedure. If there are two recursive calls, we push the second call $p(\expr_2)$ into the stack and invoke the first one $p(\expr_1)$. 
\end{compactitem}

Now we construct the super-martingale as follows. For each $i\ge 0$, we denote the stack as $\mathbf{K}_i'$ for $\mu_i'$ as $(\mathsf{func}(p),\mathsf{s}_{i,1})\cdots  (\mathsf{func}(p),\mathsf{s}_{i,q_i})$, where $q_i$ is the stack size.
We prove that another  process $y_0,y_1,\ldots$ that forms a super-martingale, where
$y_i:=\exp\left(t_*(\alpha,n^*)\cdot \left(C_i'+f(\alpha,\hat{n}_i')+\sum_{j=1}^{q_i} f(\alpha,\mathsf{s}_{i,j})\right)\right)$.
Note that $y_0 = \exp(t_*(\alpha,n^*)\cdot f(\alpha,n^*))$, and $y_{\tau'} = \exp(t_*(\alpha,n^*)\cdot C_{\tau'}')=\exp\left(t_*(\alpha,n^*)\cdot C_{\tau}\right)$.
Thus we informally have that
$
\mathbb E\left[\exp\left(t_*(\alpha,n^*)\cdot C_{\tau}\right)\right]
=\mathbb E\left[y_{\tau'} \right] \\\le \mathbb E[y_0] = \exp\left(t_*(\alpha,n^*)\cdot f(\alpha,n^*)\right)
$ and the theorem follows. 
\qed

It is natural to ask whether our theoretical approach can always find an exponential-decreasing tail bound over PRRs.
We answer this question by showing that under a difference boundedness and a monotone condition, the answer is yes. 
We first present the difference boundedness condition (A1) and the monotone condition (A2) for a PRR $\Delta$ in the canonical form (\ref{eq:canonical-prr}) as follows.

\begin{compactitem}
\item[(A1)] $\Delta$ is  \emph{difference-bounded} if there exist two real constants $M'\le M$, such that for every $n\ge c_p$, and every possible value $(V,s_1,\dots, s_k)$ in the support of the probability distribution $\etf(n,{\sf func}(p))$, we have that
\begin{talign*}
M'\cdot \mathbb E[\coststack(n)]\le V+(\sum_{i=1}^k \mathbb E[p(s_i)])-\mathbb E[p(n)] \le M\cdot \mathbb E[\coststack(n)].
\end{talign*}
\item[(A2)] $\Delta$ is \emph{expected non-decreasing} if $\mathbb E[\coststack(n)]$ does not decrease as $n$ increases.
\end{compactitem}
\noindent In other words,  (A1) says that for any possible concrete pre-processing time $V$ and passed sizes $s_1,\dots,s_k$, the difference between the expected runtime before and after the recursive call
is bounded by the magnitude of the expected pre-processing time. (A2) simply specifies that the expected pre-processing time be monotonically non-decreasing. These conditions are fulfilled by most randomized algorithms, see Appendix~\ref{app:evaluation} for details.

With the conditions (A1) and (A2), our theoretical approach guarantees a tail bound that is exponentially decreasing in the coefficient $\alpha$ and the ratio ${\mathbb E[p(n^*)]}/{\mathbb E[\coststack(n^*)]}$. The theorem statement is as follows.

\vspace{-0.3em}
\begin{theorem}\label{thm:comp}
Let $\Delta$ be a PRR in the canonical form (\ref{eq:canonical-prr}). If $\Delta$ satisfies (A1) and (A2), then for any function $w:[1, \infty)\rightarrow (1,\infty)$, the functions $f,t$ given by
\begin{talign*}
f(\alpha,n):=w(\alpha)\cdot \mathbb E[p(n)]&\mbox{ and }& t(\alpha,n):=\frac{\lambda(\alpha)}{\mathbb E[\coststack(n)]} \\
&\mbox{ with }&\lambda(\alpha):=\frac{8(w(\alpha)-1)}{w(\alpha)^2(M_2-M_1)^2}
\end{talign*}
fulfill the constraint (\ref{eq:constraint}) in Theorem \ref{thm:martingale}. Furthermore, by choosing $w(\alpha):=\frac{2\alpha}{1+\alpha}$ in the functions $f,t$ above and $\kappa(\alpha,n^*):=\alpha\cdot \mathbb E[p(n^*)]$, one obtains the tail bound
\begin{talign*}
\Pr[C_{\tau}\ge \alpha \mathbb E[p(n^*)]]\le \exp\left(-\frac{2(\alpha-1)^2}{\alpha(M_2-M_1)^2}\cdot \frac{\mathbb E[p(n^*)]}{\mathbb E[\coststack(n^*)]}\right).
\end{talign*}
\end{theorem}

\noindent\emph{Proof sketch.}
We first rephrase the constraint (\ref{eq:constraint}) as
\begin{talign*}
\mathbb E\left[\exp\left(t(\alpha,n)\cdot (\coststack(n)+\sum_{i=1}^rf(\alpha,\expr_i(n))-f(\alpha,n))\right)\right]\le 1
\end{talign*}
Then we focus on the exponent in the $\exp(\cdot)$, by (A1), the exponent is a bounded random variable. By further calculating its expectation and applying Hoeffiding's Lemma~\cite{Hoeffding1963inequality}, we obtain the theorem above.
\qed

Note that since $\mathbb E[p(n)]\ge \mathbb E[\coststack(n)]$ when $n\ge c_p$, the tail bound is at least exponentially-decreasing with respect to the coefficient $\alpha$. This implies that our theoretical approach derives tail bounds that are at least as tight as Karp's method when (A1) and (A2) holds.
When $\mathbb E[p(n)]$ is of a strictly greater magnitude than  $\mathbb E[\coststack(n)]$, our approach derives asymptotically tighter bounds.

Below, we apply the theorem above to prove tail bounds for Quickselect~(Example~\ref{ex:quickselect}) and Quicksort~(Example~\ref{ex:quicksort}).

\begin{example}\label{ex:quickselect-theory}
For QuickSelect, its canonical form is $p(n;2)=n+p(\expr_1(n))$, where $\expr_1(n)$ observes as $\mathtt{muniform}(n)$. Solving the recurrence relation, we obtain that $\mathbb E[p(n)]=4\cdot n$. We further find that this PRR satisfies (A1) with two constants  $M' = -1, M = 1$. Note that the PRR satisfies (A2) obviously. Hence, we apply Theorem \ref{thm:comp} and derive the tail bound for every sufficiently large $\alpha$:
\begin{talign*}
    \Pr[C_{\tau}\ge 4\cdot \alpha\cdot n^*]\le \exp\left(-\frac{2(\alpha-1)^2}{\alpha}\right).
\end{talign*}
On the other hand, Karp's cookbook has the tail bound
\begin{talign*}
    \Pr[C_{\tau}\ge 4\cdot \alpha\cdot n^*]\le \exp\left(1.15-1.12\cdot \alpha\right).
\end{talign*}
Our bound is asymptotically the same as Karp's but has a better coefficient.\qed
\end{example}

\begin{example}\label{ex:quicksort-theory}
For QuickSort, its canonical form is $p(n;2)=n+p(\expr_1(n))+p(\expr_2(n))$, where $\expr_1(n)$ observes as $\mathtt{muniform}(n)$ and $\expr_2(n)\! =\! n\!-\!1\!-\!\expr_1(n)$. Similar to the example above, we first calculate $\mathbb E[p(n)]\!=\!2\!\cdot\! n\!\cdot\! \ln n$. Note that this PRR also satisfies two assumptions above with two constants  $M' = -2\log 2, M = 1$. Hence, for every sufficiently large $\alpha$, we can derive the tail bound as follows:
\begin{talign*}
    Pr[C_{\tau}\ge 2\cdot \alpha\cdot n^*\cdot \ln n^*]\le \exp\left(-\frac{0.7(\alpha-1)^2}{\alpha}\cdot \ln n^*\right).
\end{talign*}
On the other hand, Karp's cookbook has the tail bound
\begin{talign*}
    Pr[C_{\tau}\ge 2\cdot \alpha\cdot n^*\cdot \ln n^*]\le \exp\left(-\alpha + 0.5\right).
\end{talign*}
Note that our tail bound is tighter than Karp's with a $\ln n$ factor.\qed
\end{example}

 {\color{red} From the generality of Markov's inequality, our theoretical approach can handle to general PRRs with three or more sub-procedure calls. However,}
the tail bounds derived from Theorem~\ref{thm:comp} is still not tight since the theorem only uses the expectation and bound of the given distribution.
For example, for QuickSelect, the tightest known bound $\exp(-{\rm \Theta}(\alpha\cdot \ln \alpha))$~\cite{QuickSelectSOTA}, is tighter than that derived from
Theorem~\ref{thm:comp}. 
Below, we present an algorithmic approach that fully utilizes the distribution information  
and derives tight tail bounds that can match ~\cite{QuickSelectSOTA}.

\vspace{-1ex}
\section{An Algorithmic Approach}\label{sec:algorithm}

In this section, we demonstrate an algorithmic implementation for our theoretical approach
(Theorem~\ref{thm:martingale}). Our algorithm
synthesizes the functions $t,f$ through template and a refined estimation on
the exponential terms {\color{blue} from the inequality (\ref{eq:constraint})}.
The estimation is via integration and the monotonicity of the template.
Below we fix a PRR $p(n;c_p)$ in the canonical form (\ref{eq:canonical-prr}) and a time limit $\alpha\cdot \kappa(n^*)$.

\smallskip
Recall that to apply Theorem~\ref{thm:martingale}, one needs to find functions $t,f$ that satisfy the constraint (\ref{eq:constraint}).
Thus, the first step of our algorithm is to have
pseudo-monomial template for
$f(\alpha,n)$ and $t(\alpha,n)$ in the following form:
\begin{small}
\begin{talign}
\label{eq:templatef} f(\alpha,n)&:=\coeff\cdot \alpha^{\polyfa} \cdot \ln^{\logfa} \alpha\cdot n^{\polyfn} \cdot \ln^{\logfn} n\\
\label{eq:templatet} t(\alpha,n)&:=\coeft\cdot \alpha^{\polyta} \cdot \ln^{\logta} \alpha\cdot n^{\polytn} \cdot \ln^{\logtn} n
\end{talign}
\end{small}

In the template, we have
$\polyfa,\logfa,\polyfn,\logfn,\polyta,\logta,\polytn,\logtn$ are given integers, and $\coeff,\coeft>0$ are unknown positive coefficients to be solved. For several compatibility reasons (see Proposition~\ref{thm:template} and~\ref{prop:failure} in the following), we require that $u_f,v_f\ge 0$ and $u_t,v_t\le 0$.
We say that the concrete values $\overline{\coeff},\overline{\coeft}$ for the unknown coefficients $\coeff,\coeft>0$ are \emph{valid} if the concrete functions $\overline{f},\overline{t}$ obtained by
substituting $\overline{\coeff},\overline{\coeft}$ for 
$\coeff,\coeft$ in the template (\ref{eq:templatef}) and (\ref{eq:templatet}) satisfy the constraint (\ref{eq:constraint}) for every sufficiently large $\alpha,n^*\ge 0$ and
all $c_p\le n\le n^*$. 

{\color{blue} We consider the pseudo-polynomial template since the runtime behavior of randomized algorithms can be mostly captured by pseudo-polynomials. We choose monomial templates since our interest is the asymptotic magnitude of the tail bound. Thus, only the monomial with the highest degrees matter.}

Our algorithm searches the values for $\polyfa,\logfa,\polyfn,\logfn,\polyta,\logta,\polytn,\logtn$ by an enumeration within a bounded range $\{-B,\dots,B\}$, where $B$ is a manually specified positive integer.
To avoid exhaustive enumeration, we use the following proposition to prune the search space.

\begin{proposition}\label{thm:template}
Suppose that we have functions $t, f: [0,\infty)\times \mathbb{N}\rightarrow [0,\infty)$ that fulfill the constraint (\ref{eq:constraint}). Then it holds that (i)
    
    $(p_f,q_f)\le (1,0)$ and $(p_t,q_t)\ge (-1,0)$, and (ii)
    
    $f(\alpha,n)={\rm \Omega}(\mathbb E[p(n)])$, $f(\alpha,n)=O(\kappa(n))$ and $t(\alpha,n)={\rm \Omega}(\kappa(n)^{-1})$ for any fixed $\alpha>0$,
where we write $(a,b)\le (c,d)$ for the lexicographic order, i.e., $(a\le c)\land (a=c\rightarrow b\le d)$.
\end{proposition}
\textit{Proof.}
Except for the constraint that $f(\alpha,n)={\rm \Omega}(\mathbb E[p(n)])$, the other constraints simply ensure that the tail bound is exponentially-decreasing. To see why $f(\alpha,n)={\rm \Omega}(\mathbb E[p(n)])$, we apply Jensen's inequality~\cite{RudinBook} to (\ref{eq:constraint}) and obtain $f(n)\ge \expv[\mathsf{Ex}(n|f)]=\expv[\coststack(n)+\sum_{i=1}^r f(\expr_i(n))]$. Then we imitate the proof of Theorem~\ref{thm:martingale} and derive that $f(n)\ge \mathbb E[p(n)]$.
\qed

{\color{blue}
Proposition~\ref{thm:template} shows that it suffices to consider (i) the choice of $u_f, v_f$ that makes the magnitude of $f$ to be within $\mathbb E[p(n)]$ and $\kappa(n)$, (ii) the choice of $u_t,v_t$ that makes the magnitude of $t^{-1}$ within $\kappa(n)$, and (iii) the choice of $p_f,q_f,p_t,q_t$ that fulfills $(p_f,q_f)\le (1,0),(p_t,q_t)\ge (-1,0)$. Note that an over-approximation of $\mathbb E[p(n)]$ can be either obtained manually or derived from automated approaches~\cite{expectedRecurrence}.}

\begin{example}\label{ex:template}
Consider the quickselect example~(Example~\ref{ex:quickselect}), suppose we are interested in the tail bound $\Pr[C_{\tau}\ge \alpha\cdot n]$, and we enumerate the eight integers in the template from $-1$ to $1$. Since $\mathbb E[p(n)] = 4\cdot n$, by the proposition above, we must have that $(u_f,v_f)=(1,0)$, $(u_t,v_t)\ge (-1,0)$, $(p_t,q_t)\ge (-1,0)$, $(p_f,q_f)\le (1,0)$. This reduces the number of choices for the template from $1296$ to $128$, where these numbers are automatically generated by our implementation. A choice is $f(\alpha,n):=c_f\cdot \alpha\cdot (\ln \alpha)^{-1}\cdot n$ and $t(\alpha,n):=c_t\cdot \ln \alpha\cdot n^{-1}$.\qed
\end{example}

In the second step, our algorithm solves the unknown coefficients $\coeft,\coeff$ in the template.
Once they are solved, our algorithm applies Theorem~\ref{thm:martingale} to obtain the tail bound.
In detail, our algorithm computes
$t_{*}(\alpha,n^*)$ as the minimum of $t(\alpha,n)$ over $c_p\le n\le n^*$, and by
$u_t,v_t\le 0$, 
$t_*(\alpha,n^*)$ is simply $t(\alpha,n^*)$, so that we obtain the tail bound $u(\alpha,n^*)=\exp(t(\alpha,n^*)\cdot (f(\alpha,n^*)-\alpha\cdot \kappa(n^*)))$.

\begin{example}\label{ex:success}
Continue with Example~\ref{ex:template}.
Suppose we have successfully found that $\overline{\coeff}=2,\overline{\coeft}=1$
is a valid concrete choice for the unknown coefficients in the template. Then $t_*(\alpha,n^*)$ is $t(\alpha,n^*)=\ln \alpha\cdot (n^*)^{-1}$, and we have the tail bound $u(\alpha,n^*)=\exp(2\cdot \alpha- \alpha\cdot \ln \alpha)$, which has better magnitude than the tail bound by Karp's method and our Theorem~\ref{thm:comp} (See Example~\ref{ex:quickselect-theory}). \qed
\end{example}

Our algorithm follows the guess-and-check paradigm.
The guess procedure 
explores possible values $\overline{\coeff},\overline{\coeft}$ for $\coeff,\coeft$ and invokes the check procedure
to verify whether the current choice is valid. Below we
present the guess procedure in Section \ref{sec:guess}, and the check procedure in Section \ref{sec:check}.

\vspace{-1ex}
\subsection{The Guess Procedure $\mathsf{Guess}(f,t)$}\label{sec:guess}

The pseudocode for our guess procedure  $\mathsf{Guess}(f,t)$ is given in Algorithm \ref{alg:main}.
In detail, it first receives a positive integer $M$ as the doubling and halving number (Line 1), then iteratively enumerates possible
values for the unknown coefficients $\coeff$ and $\coeft$ by doubling and halving for $M$ times (Line 3 -- Line 4), and finally calls the check procedure (Line 5). It is justified by the following theorem (proved in Appendix~\ref{appendix:check}).

\vspace{-0.5em}
\begin{theorem}\label{thm:binsearch}
Given the template for $f(\alpha,n)$ and $t(\alpha,n)$ as in (\ref{eq:templatef}) and (\ref{eq:templatet}), if $\overline{c_f},\overline{c_t}$ are valid choices, then (i) for every $k>1$, $k\cdot \overline{c_f},\overline{c_t}$ remains to be valid, and (ii) for every $0<k<1$, $ \overline{c_f},k\cdot \overline{c_t}$ remains to be valid.
\end{theorem}

\begin{wrapfigure}[9]{l}{0.495\textwidth}
\vspace{-2.7em}
\begin{minipage}[h]{0.495\textwidth}
\begin{scriptsize}
\begin{algorithm}[H]\label{alg:main}
\SetInd{0.3em}{0.6em}
\caption{\small{Guess Procedure}}
    \SetKwFunction{FSearch}{\textsf{Guess}}
    \SetKwProg{Fn}{Procedure}{:}{}
    \SetKwInOut{Input}{Input}
    \SetKwInOut{Output}{Output}
    \Input{Template for $f(\alpha,n)$ and $t(\alpha,n)$ as in (\ref{eq:templatef}) and (\ref{eq:templatet})}
    \Output{$\overline{\coeff},\overline{\coeft}>0$ for (\ref{eq:templatef}) and (\ref{eq:templatet})}
    \textbf{Parameter:} $M$ for the maximum steps of doubling and halving.\\
    \Fn{\FSearch{$f,t$}}{

        \For{$\overline{\coeft}:=1,2^{-1},\dots,2^{-M}$}{
\For{$\overline{\coeff}:=\frac12,1,2,\dots,2^{M-1}$}{
\If{{\sf CheckCond}$(\overline{\coeff},\overline{\coeft})$}{$\mathbf{Return}\ (\overline{\coeff},\overline{\coeft})$}
        }
        }
    }
    \smallskip
\end{algorithm}
\end{scriptsize}
\end{minipage}
\end{wrapfigure}\leavevmode

\vspace{-1em}
By Theorem~\ref{thm:binsearch}, if the check procedure is sound and complete (i.e., $\mathsf{CheckCond}$ always terminates and
$\overline{\coeff},\overline{\coeft}$ fulfills the constraint (\ref{eq:constraint}) iff $\mathsf{CheckCond}(\overline{\coeff},\overline{\coeft})$ returns true), then the guess procedure guarantees to find a solution  $\overline{\coeff},\overline{\coeft}$ (if it exists) when the parameter $M$ is large enough.
\vspace{0.5em}

\begin{example}\label{ex:guess}
    Continued with Example~\ref{ex:template}, suppose $M=2$, we enumerate $\overline{\coeff}$ from $\{\frac12,1,2\}$, and $\overline{\coeft}$ from $\{1,\frac12,\frac14\}$. We try every possible combination, and we find that  $\mathsf{CheckCond}(2,1)$ returns true. Thus, we return $(2,1)$ as the result.
    In Section~\ref{sec:check}, we will show how to conclude that $\mathsf{CheckCond}(2,1)$ is true.\qed
\end{example}

\vspace{-1ex}
\subsection{The Check Procedure $\mathsf{CheckCond}(\overline{\coeff},\overline{\coeft})$}\label{sec:check}

The check procedure takes as input the concrete values $\overline{\coeff},\overline{\coeft}$ for the unknown coefficients in the template, and outputs whether they are
valid.
It is the most involved part in our algorithm due to
the difficulty 
to
tackle the validity of the constraint~(\ref{eq:constraint}) that involves the composition of polynomials, exponentiation and logarithms.
The existence of a sound and complete decision procedure for such validity is extremely difficult and is a long-standing open problem~\cite{DBLP:conf/issac/AchatzMW08,Wilkie1997}.

To circumvent this difficulty,
the check procedure first strengthens the original constraint~(\ref{eq:constraint}) into a canonical constraint with a specific form, so that a decision algorithm that is sound and complete up to any additive error applies. Below we fix a PRR with procedure $p$ in the canonical form (\ref{eq:canonical-prr}).
We also discuss possible extensions for the check procedure in Remark~\ref{rem:extension}.

\smallskip\noindent\textbf{{The Canonical Constraint.}} We first present the canonical constraint $Q(\alpha,n)$ and how to decide the canonical constraint. The
constraint is given by (where $\forall^\infty\alpha$ means ``for all sufficiently large $\alpha$'' or formally $\exists \alpha_0. \forall \alpha\ge \alpha_0$)
\begin{small}
\begin{talign}\label{eq:canonical}
Q(\alpha,n):=\forall^\infty \alpha. \forall n\ge c_p. \left[\sum_{i=1}^{k}{\gamma_i\cdot \exp(f_i(\alpha)+g_i(n))}\le 1\right]
\end{talign}
\end{small}
subject to:
\begin{compactitem}
    \item[(C1)] For each $1\le i\le k$, we have $\gamma_i>0$ is a positive constant, $f_i(\alpha)$ is a pseudo-polynomial in $\alpha$, and $g_i(n)$ is a pseudo-polynomial in $n$. 
    \item[(C2)] For each $1\le i\le k$, the exponents for $n$ and $\ln n$ in $g_i(n)$ are non-negative. 
\end{compactitem}
We use $Q_L(\alpha,n)$ to represent the summation term
$\sum_{i=1}^{k}{\gamma_i\cdot \exp(f_i(\alpha)+g_i(n))}$
in (\ref{eq:canonical}).
Below we show that this can be checked by the algorithm {\sl Decide} up to any additive error. We present an overview of this algorithm. We also present its pseudo-code in Algorithm~\ref{alg:check-maintext}. The details are relegated into Appendix~\ref{appendix:check}.

The algorithm~{\sl Decide} requires an external function ${\sf NegativeLB}(P(n))$ that takes on input a pseudo-polynomial $P(n)$ and outputs an integer $T_n^*$ such that $P(n)\le 0$ for every $n\ge T_n^*$, or output $+\infty$ for the absence of $T_n^*$. The idea of this function is to apply the monotonicity of pseudo-polynomials. With the function ${\sf NegativeLB}(P(n))$, the algorithm~{\sl Decide} consists of two steps as follows.

\emph{First}, we can change the bound of $n$ from $[c_p,\infty)$ into $[c_p,T_n]$, where $T_n$ is a constant, without affecting the soundness and completeness. This is achieved by the observation that either: (i) we can conclude $Q(\alpha,n)$ does not hold, or (ii) there is an integer $T_n$ such that $Q_L(\alpha,n)$ is non-increasing when $n\ge T_n$. Hence, it suffices only to consider $c_p\le n\le T_n$. Below we show how to compute $T_n$ by case analysis of the limit $M_i$ of $g_i(n)$ as $n\to \infty$, for each $1\le i\le k$.
\begin{compactitem}
    \item If $M_i\!=\!+\infty$, then $\exp(g_i(n)+f_i(\alpha))$ could be arbitrarily large when $n\to \infty$.  As a result, we can conclude that $Q(\alpha,n)$ does not hold.
    \item Otherwise, by (C2), either $g_i(n)$ is a constant function, or $M_i\!=\!-\infty$. In both cases, $g_i(n)$ is non-increasing for every sufficiently large $n$. More precisely, there exists $L_i$ such that $g_i'(n)\le 0$ for every $n\ge L_i$, where $g_i'(n)$ is the derivative of $g_i(n)$. Moreover, we can invoke ${\sf NegativeLB}(g_i'(n))$ to get $L_i$.
\end{compactitem}
Finally, we set $T_n$ as the maximum of $L_i$'s and $c_p$.

\emph{Second}, for every integer $c_p\le \overline{n}\le T_n$, we substitute $n$ with $\overline{n}$ to eliminate $n$ in $Q(\alpha,n)$. Then, each exponent $f_i(\alpha)+g_i(\overline{n})$ becomes a pseudo-polynomial solely over $\alpha$. Since we only concern sufficiently large $\alpha$, we can compute the limit $R_{\overline{n}}$ for $Q_L(\alpha,\overline{n})$ as $\alpha\to \infty$. We decide based on the limit $R_{\overline{n}}$ as follows.
\begin{compactitem}
    \item If $R_{\overline{n}}<1$ for every $c_p\!\le\! \overline{n}\!\le\! L$, we conclude that $Q(\alpha,n)$ holds.
    \item If $R_{\overline{n}}\ge 1$ for some $c_p\!\le\! \overline{n}\!\le\! L$, we conclude that $Q(\alpha,n)$ does not hold to ensure soundness.
    
\end{compactitem}

\vspace{-2em}
\begin{algorithm}[htbp]\label{alg:check-maintext}
\begin{scriptsize}
\SetInd{0.3em}{0.6em}
\caption{\small{The Decision procedure for canonical constraints}}
    \label{alg:memoization}
    \SetKwFunction{FSearch}{{\sl Decide}}
    \SetKwProg{Fn}{Procedure}{:}{}
    \SetKwInOut{Input}{Input}
    \SetKwInOut{Output}{Output}
    \Input{A canonical constraint $Q(\alpha,n)$ in the form of ~(\ref{eq:canonical})}
    \Output{Decide whether $Q(\alpha,n)$ holds.}
    \Fn{\FSearch{$Q(\alpha,n)$}}{
         $T_n:=c_p$\tcp*[r]{$\triangleleft$ The first step}
        \For{$i:=1,2,\ldots ,k$}{
             $M_i:=$ The limit of $g_i(n)$ as $n\to \infty$. \\
            \If{$M_i = +\infty$}{\textbf{Return} False}
            \Else{
             $g_i'(n):=$ the derivative of $g_i(n)$\\
            $T_n := \max\{T_n, {\sf NegativeLB}(g_i'(n))\}$
            }
        }
         \For(\tcp*[f]{$\triangleleft$ The second step}){$\overline{n}:=c_p,\ldots ,T_n$}{
         $R:=0$\\
         \For{$i:=1,2,\ldots, k$}{
            $\Delta:=$ the limit of $f_i(\alpha)+g_i(\overline{n})$ as $\alpha\to \infty$.\\
             \If{$\Delta=+\infty$}{\textbf{Return} False}
             \Else{$R:=R+\gamma_i\cdot \exp(\Delta)$}
             \textbf{if}\ {$R\ge 1$}\ \textbf{then}\ {\textbf{Return} False}
         }
        }
        \textbf{Return}\ {True}
    }
\end{scriptsize}
\end{algorithm}
\vspace{-2em}

Algorithm {\sl Decide} is sound, and complete up to any additive error, as is illustrated by the following theorem. The proof is conducted directly via the definition of the limit and is relegated to Appendix~\ref{appendix:check}.
\begin{theorem}
Algorithm {\sl Decide} has the following properties:
\begin{compactitem}
       \item (Completeness) If $Q(\alpha,n)$ does not hold for infinitely many $\alpha$ and some $n\ge c_p$, then the algorithm returns false.
    \item (Soundness) For every $\varepsilon>0$, we have that if $Q_L(\alpha,n)\le 1-\varepsilon$ for all sufficiently large $\alpha$ and all $n\ge c_p$, then the algorithm returns true.
\end{compactitem}
\end{theorem}

\smallskip\noindent\textbf{The Strengthening Procedure.}
Then we show how to strengthen the constraint (\ref{eq:constraint}) into the canonical constraint (\ref{eq:canonical}), so that Algorithm {\sl Decide} applies. We rephrase (\ref{eq:constraint}) as
\begin{talign}\label{eq:rephrase}
    \expv\left[\exp(t(\alpha,n)\cdot \left(\coststack(n)+\sum_{i=1}^r f(\alpha,\expr_i(n))-f(\alpha,n)\right)\right]\le 1
\end{talign}
 and consider two functions $\overline{f},\overline{t}$ obtained by substituting the concrete values $\overline{\coeff},\overline{\coeft}$ for unknown coefficients into the template (\ref{eq:templatef}) and (\ref{eq:templatet}).
We observe that the joint-distribution of the random quantities $S(n),r\in \{1,2\}$ and $\expr_1(n),\ldots,\expr_r(n)$ in the canonical form~(\ref{eq:canonical-prr}) over PRRs can be described by several probabilistic branches  $\{c_1:B_1,\dots,c_k:B_k\}$, which corresponds to the probabilistic choice commands in the PRR. Each probabilistic branch $B_i$ has a constant probability $c_i$, a deterministic pre-processing time $S_i(n)$, a fixed number of subprocedure calls $r_i$, and a probability distribution for the variable $v$. The strengthening first handles each probabilistic branch, and then
combines the strengthening results of every branch into a single canonical constraint.

The strengthening of each branch is an application of
a set of rewriting rules. 
Intuitively, each rewriting step
over-approximates and simplifies the expectation term in the LHS of~(\ref{eq:rephrase}).
Through multiple steps of rewriting, we eventually obtain the final canonical constraint.
Below we present the details of the strengthening for a single probabilistic branch with the single recursion case. The divide-and-conquer case follows a similar treatment, see Appendix~\ref{appendix:check} for details.

Consider the single recursion case $r=1$ where a probabilistic branch has 
deterministic pre-processing time $S(n)$, distribution $\mathsf{dist}$ for the variable $v$ and passed size {\color{blue} $H(v,n)$} for the recursive call.
We have a case analysis on the distribution $\mathsf{dist}$ as follows.

\smallskip \noindent --- \emph{Case I}: $\mathsf{dist}$ is a FSDPD $\mathtt{discrete}\{c'_1:\pp_1,\dots,c'_k:\pp_k\}$, where $v$ observes as $\pp_i$ with probability $c'_i$. Then the expectation in~(\ref{eq:rephrase}) is exactly:
\begin{talign*}
    \sum_{i=1}^k c'_i\cdot \exp\left(t(\alpha,n)\cdot  S(n)+t(\alpha,n) \cdot f(\alpha,H(\pp_i,n)) - t(\alpha,n)\cdot f(\alpha,n)\right)
\end{talign*}
Thus it suffices to over-approximate the exponent $X_i(\alpha,n):=t(\alpha,n)\cdot  \coststack(n)+t(\alpha,n) \cdot f(\alpha,H(\pp_i,n)) - t(\alpha,n)\cdot f(n)$ into the form subject to (C1)--(C2). {\color{blue} For this purpose, our strengthening repeatedly applies the following rewriting rules (R1)--(R4) for which $0<a<1$ and $b>0$:}  
\begin{footnotesize}
\begin{talign*}\vspace{-1em}
    &\text{(R1) }f(\alpha,H(\pp_i,n))\le f(\alpha,n)\quad \\ &\text{(R2) }\ln (an-b)\le \ln n+\ln a\quad \ln (an+b)\le \ln n+\ln (\min\{1,a+\frac{b}{c_p}\})   \\
    &\text{(R3) }0\le n^{-1}\le c_p^{-1}\quad 0\le \ln^{-1}n\le \ln^{-1} c_p
    \quad \text{(R4) } \lfloor \frac nb\rfloor \le \frac nb\quad \lceil \frac nb\rceil \le \frac nb+\frac{b-1}{b}
\end{talign*}
\end{footnotesize}
\noindent (R1) follows from the well-formedness $0\le H(\expr_i,n)\le n$ and the monotonicity of $f(\alpha,n)$ with respect to $n$. (R2)--(R4) are straightforward. Intuitively, (R1) can be used to cancel the term $f(\alpha,H(\expr_i,n))-f(\alpha,n)$, (R2) simplifies the sub-expression in $\ln$, (R3) is used to remove floors and ceils, and (R4)
to remove $n^{-c}$ and $\ln^{-c}n$ to satisfy the restriction (C2) of the canonical constraint.
{\color{blue} To apply these rules, we consider two strategies below.}
\begin{compactitem}
    \item[(S1-D)] Apply (R1) and over-approximate $X_i(\alpha,n)$ as $t(\alpha,n)\cdot \coststack(n)$. Then, we repeatedly apply (R3) to remove terms $n^{-c}$ and $\ln^{-c}n$.
    \item[(S2-D)] Substitute $f$ and $t$  with the concrete functions $\overline{f},\overline{t}$ and expand $H(\pp_i,n)$.
    Then we first apply (R4) to remove all floors and ceils, and repeatedly apply (R2) to replace all occurrences of $\ln(an+b)$ with $\ln n+\ln C$ for some constant $C$. By the previous replacement, the whole term $X_i(\alpha,n)$ will be over-approximated as a pseudo-polynomial over $\alpha$ and $n$. Finally, we eagerly apply (R3) to remove all terms $n^{-c}$ and $\ln^{-c}n$.
\end{compactitem}
{\color{red} Our algorithm first tries to apply (S2-D), if it fails to derive a canonical constraint, then we apply the alternative (S1-D) to the original constraint.}
If both the strategies fails,
we report failure and exit the check procedure.
\begin{example}
    Suppose $v$ observes as $\{0.5:n-1,0.5:n-2\}, S(n):=\ln n, t(\alpha,n):=\frac{\ln \alpha}{\ln n}, f(\alpha,n):=4\cdot \frac{\alpha}{\ln\alpha} \cdot n\cdot \ln n, H(v,n):=v$. We consider applying both strategies to the first term $\pp_1:=n-1$ and $X_1(\alpha,n):=t(\alpha,n)\cdot (\coststack(n)+ f(\alpha,n-1) -f(\alpha,n))$.
    If we apply (S1-D) to $X_1$, it will be approximated as $\exp(\ln \alpha)$.
    If we apply (S2-D) to $X_1$, it will be first over-approximated as $\frac{\ln \alpha}{\ln n}\cdot (\ln n+{4\cdot \frac{\alpha}{\ln\alpha} \cdot v\cdot \ln n}-4\cdot \frac{\alpha}{\ln\alpha}\cdot n\cdot \ln n)$, then we substitute $v=n-1$ and derive the final result $\exp(\ln \alpha-4\cdot \alpha)$. {\color{blue} Hence, both the strategies succeed.} \qed 
\end{example}

\smallskip \noindent --- \emph{Case II}: $\mathsf{dist}$ is $\mathtt{uniform}(n)$ or $\mathtt{muniform}(n)$. Note that $H(v,n)$ is linear with respect to $v$, thus $H(v,n)$ is a bijection over $v$ for every fixed $n$.
Hence, if $v$ observes as  $\mathtt{uniform}(n)$, then  
\begin{small}
\begin{talign}\label{eq:sum}
    \expv[\exp(t(\alpha,n)\cdot f(\alpha,H(v,n))) ]\le \frac1n \sum_{v=0}^{n-1} \exp(t(\alpha,n)\cdot f(\alpha,v))
\end{talign}
\end{small}
If $v$ observes as $\mathtt{muniform}(n)$, a similar inequality holds by replacing $\frac1n$ with $\frac2n$.
Since $f(\alpha,v)$ is a non-decreasing function with respect to $v$, we further over-approximate the summation in (\ref{eq:sum}) by the integral $\int_{0}^{n}\exp(t(\alpha,n)\cdot f(\alpha,v))\text{d}v$.

\begin{example}\label{ex:int}
    Continue with Example~\ref{ex:guess}, we need to check
    
    $\overline{t}(\alpha,n)=\frac{\ln \alpha}{n}$ and $\overline{f}(\alpha,n)=\frac{2\cdot \alpha
    }{\ln \alpha}\cdot n$. By the inequality~(\ref{eq:sum}), we expand the constraint~(\ref{eq:rephrase}) into $\frac{2}{n}\cdot \exp(\ln \alpha-2\cdot \alpha)\cdot \sum_{v=0}^{n-1}\exp(\frac{2\cdot \alpha\cdot i}{n})$. By integration, it is further over-approximated as $\frac{2}{n}\cdot \exp(\ln \alpha-2\cdot \alpha)\cdot \int_{0}^{n}\exp(\frac{2\cdot \alpha\cdot v}{n})\text{d}v$.\qed
\end{example}
Note that we still need to resolve the integration of an exponential function whose exponent is a pseudo-monomial over $\alpha,n,v$. Below we denote by $d_v$ the degree on the variable $v$ and by $\ell_v$ the degree of $\ln v$. We first  list the situations where the integral can be computed exactly.
\begin{compactitem}
\item If $(d_v,\ell_v)=(1,0)$, then the exponent could be expressed as $W(\alpha,n)\cdot v$,where $W(\alpha,n)$ is a pseudo-monomial over $\alpha$ and $n$. We can compute the integral as $\frac{\exp(n\cdot W(\alpha,n))-1}{W(\alpha,n)}$ and over-approximate it as
 $\frac{\exp(n\cdot W(\alpha,n))}{W(\alpha,n)}$ by removing $-1$ in the numerator.
\item If $(d_v,\ell_v)=(0,1)$, then the exponent is of the form $W(\alpha,n)\cdot \ln v$. We follow a similar procedure with the case above and obtain the over-approximation $\frac{n \cdot \exp(\ln n\cdot W(\alpha,n))}{W(\alpha,n)}$.
\item If $(d_v,\ell_v) = (0,0)$, then the result is trivially $n\cdot \exp(W(\alpha,n))$.
\end{compactitem}
Then we handle the situation where the exact computation of the integral is infeasible.
In this situation, the strengthening further over-approximates the integral into simpler forms by
{first replacing $\ln v$ with $\ln n$, and then replacing $v$ with $n$ to reduce the degrees $\ell_v$ and $d_v$.}
Eventually, the exponent in the integral bows down to
{one of the three situations (where the integral can be computed exactly) above, and the strengthening returns the exact value of the integral.}

\begin{example}\label{ex:approx-int}
    Continue with Example~\ref{ex:int}. We express the exponent as $\frac{2\cdot \alpha}{n}\cdot v$. Thus, we can plug $\frac{2\cdot \alpha}{n}$ into $W(\alpha,n)$ and obtain the integration result $\frac{\exp(2\cdot \alpha)}{2\cdot \alpha/n}$. Furthermore, we can simplify the formula in Example~\ref{ex:int} as $\frac{\exp(\ln \alpha)}{\alpha}$.\qed
\end{example}

In the end, we move the term $\frac1n$ (or $\frac2n$) that comes from the \texttt{uniform} (or \texttt{muniform}) distribution and the coefficient term $W(\alpha,n)$ into the exponent.
If we move these terms directly, it may produce $\ln \ln n$ and $\ln \ln \alpha$ that comes from taking the logarithm of $\ln n$ and $\ln \alpha$.
Hence, we first apply $\ln c_p\le \ln n\le n$ and $1\le \ln \alpha\le \alpha$ to remove all terms $\ln n$ and $\ln \alpha$ outside the exponent (e.g., $\frac{\ln \alpha}{\ln n}$ is over-approximated as $\frac{\alpha}{\ln c_p}$). After the over-approximation, the terms outside the exponentiation form a polynomial over $\alpha$ and $n$, we can trivially move these terms into the exponent by taking the logarithm. Finally, we apply (R4) in Case I to remove $n^{-c}$ and $\ln^{-c} n$. 
If we fail to obtain the canonical constraint, the strengthening reports failure.

\begin{example}\label{ex:final}
    Continue with Example~\ref{ex:approx-int},
    we move the term $\alpha$ into the exponentiation and simplify the over-approximation result as ${\exp(\ln \alpha-\ln \alpha)}=1$. As a result, we over-approximate the LHS of (\ref{eq:rephrase}) as $1$ and we conclude that $\mathsf{CheckCond}(2,1)$ holds. \qed
\end{example}

The details of the divide-and-conquer case are similar and relegated to Appendix~\ref{appendix:check}. Furthermore, we present how to combine the strengthening results for different branches into a single canonical constraint. Suppose for every probabilistic branch $B_i$, we have successfully obtained the canonical constraint $Q_{L,i}(\alpha,n)\le 1$ as the strengthening of the original constraint~(\ref{eq:rephrase}). Then, the canonical constraint for the whole distribution is $\sum_{i=1}^{k}c_i\cdot Q_{L,i}(\alpha,n)\le 1$.
Intuitively, there is probability $c_i$ for the branch $B_i$, thus the combination follows by simply expanding the expectation term.

A natural question is to ask whether our algorithm can always succeed to obtain the canonical constraint. We have the proposition as follows.
\begin{proposition}\label{prop:failure}
    If the template for $t$ has a lower magnitude than $S(n)^{-1}$ for every branch, then the rewriting always succeeds.
\end{proposition}
\begin{proof}
We first consider the single recursion case.
    When $\mathsf{dist}$ is FSDPD, we can apply (S1-D) to over-approximate the exponent as $t(\alpha,n)\cdot S(n)$. Since $t(\alpha,n)$ has a lower magnitude than $S(n)^{-1}$, by further applying (R3) to eliminate $n^{-c}$ and $\ln^{-c}n$, we obtain the canonical constraint.
    If $\mathsf{dist}$ is $\mathtt{uniform}(n)$ or $\mathtt{muniform}(n)$
    , we observe that the over-approximation result for the integral is either $\frac{\exp(f(\alpha,n))}{ f(\alpha,n)\cdot t(\alpha,n)}$ (when $d_v>0$) or $\frac{\ln n\cdot \exp(f(\alpha,n))}{f(\alpha,n)\cdot t(\alpha,n)}$ (when $d_v=0$). Thus, we can cancel the term $f(\alpha,n)$ in the exponent and obtain the canonical constraint by the subsequent steps.
    The proof is the same for the divide-and-conquer case.
\qed
\end{proof}
By Proposition~\ref{prop:failure}, we restrict
$u_t,v_t\le 0$ in the template to ensure our algorithm never fails.

{\color{green}
\begin{remark}\label{rem:extension}
Our algorithm can be extended to support piecewise uniform distributions (e.g.  each of $0,\dots, n/2$ with probability $\frac{2}{3n}$ and each of $n/2+1,\dots, n-1$ with probability $\frac{4}{3n}$) by handling each piece separately.
\end{remark}}

\vspace{-2ex}
\section{Experimental Results}
\label{sec:evaluation}

In this section, 
we evaluated our algorithm over classical randomized algorithms such as {QuickSort} (Example~\ref{ex:quicksort}), {QuickSelect} (Example~\ref{ex:quickselect}), {DiameterComputation}~\cite[Chapter 9]{DBLP:books/cu/MotwaniR95}, {RandomizedSearch}~\cite[Chapter~9]{McConnell},
{ChannelConflictResolution}~\cite[Chapter~13]{DBLP:books/daglib/0015106},
examples such as Rdwalk and Rdadder in the literature~\cite{SriramCAV}, and four manually-crafted examples (MC1 -- MC4).
For each example, we manually compute its expected running time
for the prunning.
A detailed description of all these examples is provided in Appendix~\ref{app:examples}.

\begin{wrapfigure}[12]{r}{0.4\textwidth}
\begin{center}
\vspace{-2.4em}
\includegraphics[width=\linewidth]{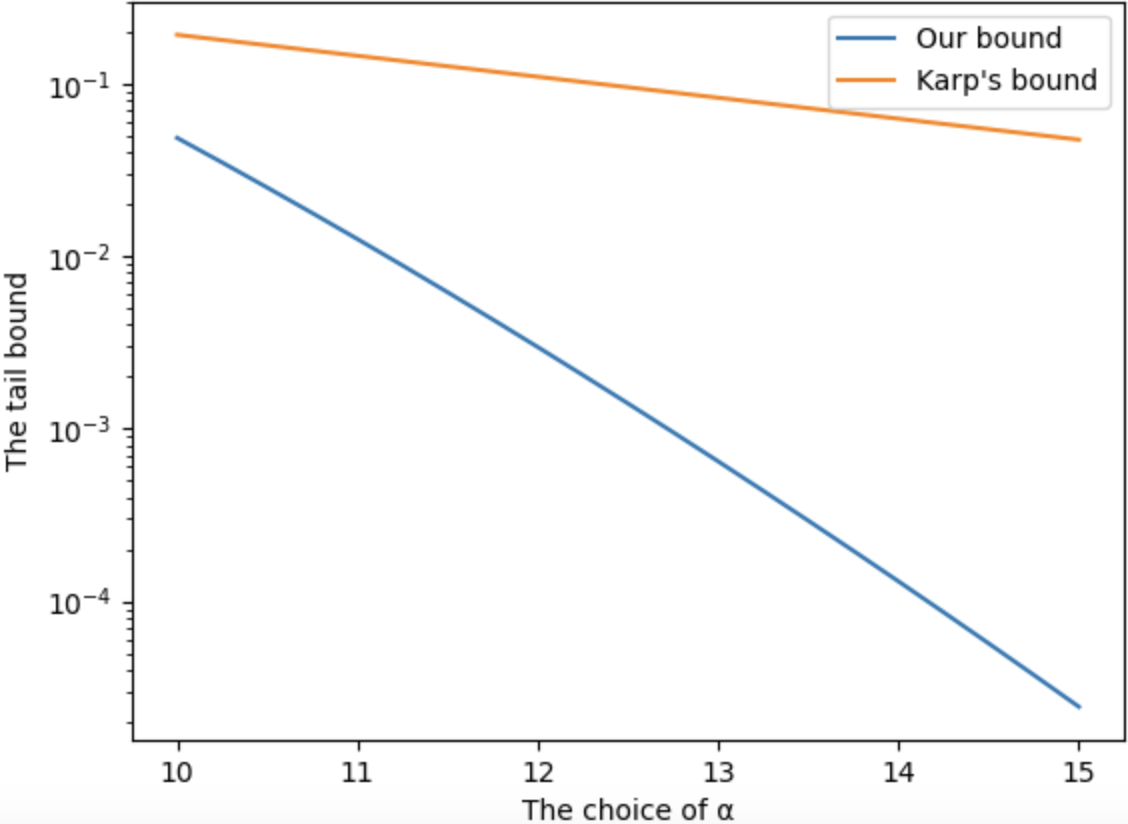}
\caption{Plot for QuickSelect}
\label{fig:plot-maintext}
\end{center}
\end{wrapfigure}

We implemented our algorithm in \texttt{C++}. We choose $B=2$ (as the bounded range for the template), $M=4$ (in the guess procedure), $Q=8$ (for the number of parts in the integral),
and prune the search space by Theorem~\ref{thm:template}.
All results were obtained on an Ubuntu 18.04 machine with an 8-Core Intel i7-7900x Processor (4.30 GHz) and 40 GB of RAM.

We report the tail bound derived by our algorithm in Table~\ref{table:exp}, where ``Benchmark'' lists the benchmarks, ``$\alpha\cdot \kappa(n^*)$'' lists the time limit of interest, ``Our bound'' lists the tail bound by our approach, ``Time(s)'' lists the runtime (in seconds) of our approach, and ``Karp's bound'' lists the bounds by Karp's method. From the table, our algorithm constantly derives asymtotically tighter tail bounds than Karp's method. Moreover, all these bounds are obtained in a few seconds, demonstrating the efficiency of our algorithm. Furthermore, our algorithm obtains bounds with tighter magnitude than our completeness theorem (Theorem~\ref{thm:comp}) in 9 benchmarks, and 
bounds with the same magnitude as the others, see Appendix~\ref{app:evaluation} for details.

For an intuitive comparison, we also report the concrete bounds and their plots of our method and Karp's method. We choose three concrete choices of $\alpha$ and $n^*$ and plot the concrete bounds over $10\le \alpha\le 15,n^*=17$. For concrete bounds, we also report the ratio $\frac{\text{Karp's Bound}}{\text{Our Bound}}$ to show the strength of our method. Due to space limitations, we only report the results for QuickSelect (Example~\ref{ex:quickselect}) in Table~\ref{fig:concrete-maintext} and Figure~\ref{fig:plot-maintext}, while the rest are relegated to Appendix~\ref{app:evaluation}.

\begin{table}[htbp]
\renewcommand{\baselinestretch}{1.3}
\vspace{-1ex}
\begin{center}
\begin{scriptsize}
    \caption{Experimental Result}
    \label{table:exp}
    \begin{tabular}{|c|c|c|c|c|}
    \hline
        Benchmark & $\alpha\cdot \kappa(n^*)$ in~(\ref{eq:problem}) & Our bound & Time(s) & Karp's bound\\
    \hline
        QuickSelect & $\alpha\cdot n^*$ & $ \exp(2\cdot \alpha-\alpha\cdot \ln\alpha)$ & $0.03$ & $\exp(1.15-0.28\cdot \alpha)$  \\
    \hline
        QuickSort &$ \alpha\cdot n^*\cdot \ln n^*$ & $\exp((4-\alpha)\cdot \ln n^*)$ & $0.02$ & $\exp(0.5-0.5\cdot \alpha)$ \\
    \hline
        L1Diameter & $\alpha\cdot n^*$  & $\exp(\alpha-\alpha\cdot \ln\alpha)$ & $0.03$ & $\exp(1.39-0.69\cdot
        \alpha)$ \\
    \hline
        L2Diameter & $\alpha\cdot n^*\cdot \ln n^*$  & $\exp(\alpha-\alpha\cdot \ln\alpha)$ & $0.03$ & $\exp(1.39-0.69\cdot
        \alpha)$  \\
    \hline
        RandSearch & $\alpha\cdot \ln n^*$ & $\exp((2\cdot \alpha-\alpha\cdot \ln\alpha)\cdot \ln n^*)$ & $0.03$ & $\exp(-0.29\cdot \alpha\cdot \ln n^* )$  \\
    \hline
        Channel & $\alpha\cdot n^*$ & $\exp((8-\alpha)\cdot n^*)$ & $0.05$ & $\exp(1-0.37\cdot \alpha)$  \\
    \hline
        Rdwalk & $\alpha\cdot n^*$  & $\exp((0.5-\alpha)\cdot n^*)$ & $0.05$ & $\exp(0.60-0.41\cdot\alpha)$ \\
    \hline
        Rdadder & $\alpha\cdot n^*$ & $\exp((4-0.5\cdot \alpha)\cdot n^*)$ & $0.04$ & Not applicable \\
    \hline
        MC1 & $\alpha\cdot \ln n^*$  & $\exp((\alpha-\alpha\cdot \ln \alpha)\cdot \ln n^*)$ & $0.03$ & $\exp(-0.69\cdot
        \alpha\cdot \ln n^*)$ \\
\hline
        MC2 &$\alpha\cdot \ln^2 n^*$ & $\exp((\alpha - \alpha\cdot \ln \alpha)\cdot \ln n^*)$ & $0.03$& $\exp(-0.69\cdot \alpha\cdot \ln n^*)$  \\
    \hline
        MC3 & $\alpha\cdot n^*\cdot \ln^2 n^*$  & $\exp(\alpha-\alpha\cdot \ln\alpha)$ & $0.03$ & $\exp(1.15-0.28\cdot
        \alpha)$ \\
    \hline
        MC4 &$\alpha\cdot n^*$ & $\exp(2\cdot \alpha - \alpha\cdot \ln \alpha)$ & $0.04$ & Not applicable  \\
    \hline
\end{tabular}
\end{scriptsize}
\end{center}
\vspace{-2ex}
\end{table}

\vspace{-4em}

\begin{table}
\begin{center}
    
\caption{Concrete Bounds for QuickSelect}
\label{fig:concrete-maintext}
\begin{scriptsize}
    \begin{tabular}{|c|c|c|c|c|}
    \hline
        Concrete choice & Our bound & Karp's Bound & Ratio \\ \hline
        $\alpha=10;n^*=13$ & $0.0485$ & $0.192$ & $3.96$ \\ \hline
        $\alpha=11;n^*=15$ & $0.0126$ & $0.145$ & $11.6$ \\ \hline
        $\alpha=12;n^*=17$ & $0.00297$ & $0.110$ & $36.9$ \\ \hline
    \end{tabular}
\end{scriptsize}
\end{center}
\end{table}

\vspace{-4em}
\section{Related Work}
\label{sec:related}

\noindent\textbf{\textit{{Karp's Cookbook}}}
Our approach is orthogonal to Karp's cookbook method~\cite{DBLP:journals/jacm/Karp94} since we base our approach on Markov's inequality, and the core of Karp's method is a dedicated proof for establishing that an intricate tail bound function is a prefixed point of the higher order operator derived from the given PRR. Furthermore, our automated approach can derive asymptotically tighter tail bounds than Karp's method over all 12 PRRs in our benchmark. Our approach could also handle randomized preprocessing times, which is beyond the reach of Karp's method. Since Karp's proof of prefixed point is ad-hoc, it is non-trivial to extend his method to handle the randomized cost.
Nevertheless, there are PRRs (e.g., Coupon-Collector, see Appendix~\ref{app:notfixedpoint}) that can be handled by Karp's method but not by ours. Thus, our approach provides a novel way to obtain asymptotically tighter tail bounds than Karp's method. 

The recent work~\cite{DBLP:conf/itp/Tassarotti018} extends Karp's method for deriving tail bounds for parallel randomized algorithms. This method derives the same tail bounds as Karp's method over PRRs with a single recursive call (such as QuickSelect) and cannot handle randomized pre-processing time. Compared with this approach, our approach derives tail bounds with tighter magnitude on 11/12(92.67\

\noindent\textbf{\textit{{Custom Analysis.}}} Custom analysis of PRRs~\cite{QuickSelectSOTA,DBLP:journals/jal/McDiarmidH96} has successfully derived tight tail bounds for QuickSelect and QuickSort. Compared with the custom analysis that requires ad-hoc proofs, our approach is automated, has the generality from Markov's inequality, and is capable of deriving bounds identical or very close to the tail bounds from the custom analysis.

\noindent\textbf{\textit{Probabilistic Programs.}} There are also relevant approaches in probabilistic program verification. These approaches are either based on martingale concentration inequalities (for exponentially-decreasing tail bounds)~\cite{SriramCAV,ChatterjeeNZ2017,DBLP:journals/toplas/ChatterjeeFNH18,DBLP:conf/aplas/HuangFC18,DBLP:conf/cav/ChatterjeeGMZ22}, Markov's inequality (for polynomially-decreasing tail bounds)~\cite{DBLP:journals/corr/ChatterjeeF17,DBLP:conf/tacas/KuraUH19,pldi21wang},  fixed-point synthesis~\cite{pldi21ours}, or weakest precondition reasoning~\cite{DBLP:journals/jacm/KaminskiKMO18,poplAmortized}. Compared with these approaches, our approach is dedicated to PRRs (a light-weight representation of recursive probabilistic programs) and involves specific treatment of common recursive patterns (such as randomized pivoting and divide-and-conquer) in randomized algorithms, while these approaches usually do not consider common recursion patterns in randomized algorithms. Below we have detailed technical comparisons with these approaches.

\noindent -- Compared with the approaches based on martingale concentration inequalities~\cite{SriramCAV,DBLP:conf/cav/ChatterjeeGMZ22,ChatterjeeNZ2017,DBLP:journals/toplas/ChatterjeeFNH18,DBLP:conf/aplas/HuangFC18}, our approach has the same root as them, since martingale concentration inequalities are often proved via Markov's inequality. However, those approaches have more accuracy loss since these martingale concentration inequalities usually make further relaxations after applying Markov's inequality. In contrast, our automated approach directly handles the constraint after applying Markov's inequality by having a refined treatment of exponentiation and hence has better accuracy in deriving tail bounds.

\noindent -- Compared with the approaches~\cite{DBLP:journals/corr/ChatterjeeF17,DBLP:conf/tacas/KuraUH19,pldi21wang} that derive polynomially-decreasing tail bounds, 
our approach targets the sharper exponentially-decreasing tail bounds and hence is orthogonal.

\noindent -- Compared with the fixed-point synthesis approach~\cite{pldi21ours}, our approach is orthogonal as it is based on Markov's inequality. Note that the approach~\cite{pldi21ours} can only handle 3/12 (25\

\noindent -- Compared with weakest precondition reasoning~\cite{DBLP:journals/jacm/KaminskiKMO18,poplAmortized} that requires first specifying the bound functions and then verifying the bound functions by proof rules related to fixed-point conditions, mainly with manual efforts, our approach can be automated and is based on Markov's inequality rather than fixed point theorems. Although Karp's method is also based on a particular tail bound function as a prefixed point and can thus be embedded into the weakest precondition framework, Karp's proof of prefixed point requires deep insight, which is beyond existing proof rules. Moreover, even a slight relaxation of the tail bound function into a simpler form in Karp's method no longer keeps the bound function to be a prefixed point (see Appendix~\ref{app:notfixedpoint} for details). Hence, the approach of the weakest precondition may not be suitable for deriving tail bounds.

\section*{Acknowledgement}
\vspace{-0.5em}
We thank Prof. Bican Xia for valuable information on the exponential theory of reals.

 \bibliographystyle{splncs04}
 \bibliography{PL.bib}

 \clearpage
\appendix
\section{An Illustrative Example of a Martingale}\label{app:martingales}

\begin{example}
	Consider an unbiased and discrete random walk, in which we start at a position $X_0$, and at each second walk one step to either left or right with equal probability. Let $X_n$ denote our position after $n$ seconds. It is easy to verify that $\expv\left(X_{n+1} \vert X_0, \ldots, X_n\right) = \frac{1}{2} (X_n - 1) + \frac{1}{2} (X_n + 1) = X_n.$ Hence, this random walk is a martingale. Note that by definition, every martingale is also a supermartingale.
	As another example, consider the classical gambler's ruin: a gambler starts with $Y_0$ dollars of money and bets continuously until he loses all of his money. If the bets are unfair, i.e.~the expected value of his money after a bet is less than its expected value before the bet, then the sequence $\{Y_n\}_{n \in \mathbb{N}_0}$ is a supermartingale. In this case, $Y_n$ is the gambler's total money after $n$ bets. On the other hand, if the bets are fair, then $\{Y_n\}_{n \in \mathbb{N}_0}$ is a martingale. See Appendix~\ref{app:martingales} for another illustrative example of martingales. \qed
\end{example}

\begin{example}[P\'olya's Urn~\cite{mahmoud2008polyas}]
	As a more interesting example, consider an urn that initially contains $R_0$ red and $B_0$ blue marbles ($R_0+B_0>0$). At each step, we take one marble from the urn, chosen uniformly at random, look at its color and then add two marbles of that color to the urn. Let $B_n, R_n$ and $M_n$ respectively be the number of red, blue and all marbles after $n$ steps. Also, let $\beta_n = \frac{B_n}{M_n}$ and $\rho_n = \frac{R_n}{M_n}$ be the proportion of marbles that are blue (resp. red) after $n$ steps. Let $\mathcal{F}_n$ model the observations until the $n$-th step. The process described above leads to the following equations:
	$$
	M_{n+1} = 1 + M_n,
	$$
	$$
	\condexpv{B_{n+1}}{\mathcal{F}_n} = \condexpv{B_{n+1}}{B_1, \ldots, B_n} = \frac{B_n}{M_n} \cdot (B_n + 1) + \frac{R_n}{M_n} \cdot B_n,
	$$
	$$
	\condexpv{R_{n+1}}{\mathcal{F}_n}=\condexpv{R_{n+1}}{B_1, \ldots, B_n} = \frac{R_n}{M_n} \cdot (R_n + 1) + \frac{B_n}{M_n} \cdot R_n.
	$$
	Note that we did not need to care about observing $R_i$'s, $M_i$'s, $\beta_i$'s or $\rho_i$'s, because they can be uniquely computed in terms of $B_i$'s. More generally, an observer can observe only $B_i$'s, or only $R_i$'s, or only $\beta_i$'s or $\rho_i$'s and can then compute the rest using this information. Based on the equations above, we have:
	$$
	\condexpv{\beta_{n+1}}{\mathcal{F}_n} = \frac{B_n}{M_n} \cdot \frac{B_{n}+1}{M_n + 1} + \frac{M_n-B_n}{M_n} \cdot \frac{B_n}{M_n + 1} = \frac{B_n}{M_n} = \beta_n,
	$$
	$$
	\condexpv{\rho_{n+1}}{\mathcal{F}_n} = \frac{R_n}{M_n} \cdot \frac{R_{n}+1}{M_n + 1} + \frac{M_n-R_n}{M_n} \cdot \frac{R_n}{M_n + 1} = \frac{R_n}{M_n} = \rho_n.
	$$
	Hence, both $\{\beta_n\}_{n \in \mathbb{N}_0}$ and $\{\rho_n\}_{n \in \mathbb{N}_0}$ are martingales. Informally, this means that the expected proportion of blue marbles in the next step is exactly equal to their observed proportion in the current step. This might be counter-intuitive. For example, consider a state where 0.99 of the marbles are blue. Then, it is more likely that we will add a blue marble in the next state. However, this is mitigated by the fact that adding a blue marble changes the proportions much less dramatically than adding a red marble.
\end{example}

\section{Examples of Probabilistic Recurrence Relations}\label{app:examples}

\begin{example}[{QuickSelect}]
	\label{ex:quickselect-appendix}
	Consider the problem of finding the $d$-th smallest element in an unordered array of $n$ distinct elements. A classical randomized algorithm for solving this problem is {QuickSelect}~\cite{DBLP:journals/cacm/Hoare61a}. 
It begins by choosing a pivot element $u$ of the array uniformly at random. It then compares all the other $n-1$ elements of the array with $u$ and divides them into two parts: (i)~those that are smaller than $u$ and (ii)~those that are larger. Suppose that there are $d'$ elements in part (i). If $d'<d-1,$ then the algorithm recursively searches for the $(d-d'-1)$-th smallest element of part (ii). If $d'=d-1,$ the algorithm terminates by returning $u$ as the desired answer. Finally, if $d'>d-1,$ the algorithm recursively finds the $d-$th smallest element in part (i).
Note that the classical median selection algorithm is a special case of {QuickSelect}.
While there is an involved deterministic algorithm for solving the same problem, more involved linear-time non-randomized algorithms exist for the same problem,
{QuickSelect}
provides a simple randomized variant matching the best-known $O(n)$ deterministic runtime. Since there is only one recursive procedure in {QuickSelect}, we model the algorithm as the following PRR.
	$$\textstyle
\mathtt{def}\ p(n;2)= \{\mathtt{sample}\ v\leftarrow\mathtt{muniform}(n)\ \mathtt{in}\ \{\kwtick(n);\ \mathtt{invoke}\ p(v);\}\}
$$

Here, $p(n)$ represents the number of comparisons performed by {QuickSelect} over an input of size $n$, and $v$ is the random variable that captures the size of the remaining array that has to be searched recursively. It can be derived that $v = \max\{i, n-i-1\}$ where $i$ is sampled uniformly from $\{0, \ldots, n-1\}.$ \qed

\end{example}

\begin{example}[{QuickSort}]
	Consider the classical problem of sorting an array of $n$ distinct elements.
A well-known randomized algorithm for solving this problem is {QuickSort}.
Similar to {QuickSelect}, {QuickSort} begins by choosing a pivot element $u$ of the array uniformly at random. It then compares all the other $n-1$ elements of the array with $u$ and divides them into two sub-arrays: (i)~those that are smaller than $u$ and (ii)~those that are larger. The algorithm then runs recursively on both sub-arrays.
Compared with its deterministic counterpart, {MergeSort},
{QuickSort} provides a simple randomized variant whose expected runtime matches the best-known deterministic runtime $O(n\log n)$.

We model the algorithm as the following PRR:
$$
\mathtt{def}\ p(n;2)= \{\mathtt{sample}\ v\leftarrow\mathtt{uniform}(n)\ \mathtt{in}\ \{\kwtick(n);\ \mathtt{invoke}\ p(v);p(n-1-v);\}\}
$$
Here, $v$ and $n-1-v$ are random variables that capture the sizes of the two sub-arrays, where $i$ is uniformly sampled from $\{0, \ldots, n-1\}.$
\qed
\end{example}

\begin{example}[{DiameterComputation}]\label{ex:diameter}
Consider the {DiameterComputation} algorithm~\cite[Chapter 9]{DBLP:books/cu/MotwaniR95} to compute the diameter of an input finite set $S$ of three-dimensional points.
Depending on the metric, i.e.~Eucledian or $L_1,$ we obtain two different recurrence relations.
For the Euclidean $L_2$ metric, we have the following PRR:
$$\textstyle
\mathtt{def}\ p(n;2)= \{\mathtt{sample}\ v\leftarrow\mathtt{uniform}(n)\ \mathtt{in}\ \{\kwtick(n\cdot \ln n);\ \mathtt{invoke}\ p(v);\}\}
$$
For the $L_1$ metric, the PRR is as follows:
$$\textstyle
\mathtt{def}\ p(n;2)= \{\mathtt{sample}\ v\leftarrow\mathtt{uniform}(n)\ \mathtt{in}\ \{\kwtick(n);\ \mathtt{invoke}\ p(v);\}\}
$$
Note that here we use deterministic versions for the subroutine {HalfspaceIntersection} as required in {DiameterComputation}, see~\cite{DBLP:journals/tcs/PreparataM79} for the Euclidean case and~\cite[Problem 9.6, Page 276]{DBLP:books/cu/MotwaniR95} for the $L_1$ case.
\qed
\end{example}

\begin{example}[{RandomizedSearch}]\label{ex:randsearch}
Consider Sherwood's {RandomizedSearch} algorithm (cf.~\cite[Chapter~9]{McConnell}).
The algorithm checks whether an integer value $d$ is present within the index range $[i,j]$ ($0\le i\le j$)
in an integer array $ar$ which is sorted in increasing order and is without duplicate entries.
The algorithm outputs either the index of the item or $-1$ if $d$ is not present in
the index range $[i,j]$.
The PRR for this example contains only one recurrence equation*:
$$
\mathtt{def}\ p(n;2)= \{\mathtt{sample}\ v\leftarrow\mathtt{muniform}(n)\ \mathtt{in}\ \{\kwtick(1);\ \mathtt{invoke}\ p(v);\}\}
$$
\qed
\end{example}

\begin{example}[{ChannelConflictResolution}]\label{ex:channel}
We consider two network scenarios in which $n$ clients are trying to get
access to a network channel.
This problem is also called Resource-Contention Resolution~\cite[Chapter~13]{DBLP:books/daglib/0015106}.
In this problem, if more than one client tries to access the channel,
then no client can access it, and if exactly one client requests access to
the channel, then the request is granted.
In the concurrent setting, the clients share one variable, which is the number
of clients which have not yet been granted access.
Also in this scenario, once a client gets an access the client does not
request for access again.
For this problem, we obtain an over-approximating recurrence relation
\begin{equation*}
\mathtt{def}\ p(n;2)= \bigoplus \begin{cases}
\frac{1}{e}: \mathtt{pre}(1);  \mathtt{invoke}\  p(n-1);\\
1-\frac{1}{e}: \mathtt{pre}(1);  \mathtt{invoke}\  p(n);
\end{cases}
\end{equation*}
\end{example}

\begin{example}[RandomWalk]\label{ex:rdwalk}
We consider the classic random walk cases appears in many previous literatures~\cite{ChatterjeeFG16,ChatterjeeNZ2017,SriramCAV,pldi21ours},
\begin{equation*}
\mathtt{def}\ p(n;1)= \bigoplus \begin{cases}
0.5: \mathtt{pre}(1);  \mathtt{invoke}\  p(n);\\
0.5: \mathtt{pre}(1);  \mathtt{invoke}\  p(n-3);
\end{cases}
\end{equation*}
\end{example}

\begin{example}[RandomAdder]\label{ex:rdwalk}
We consider the random accumulator from previous literatures~\cite{SriramCAV,pldi21ours},
\begin{equation*}
\mathtt{def}\ p(n;1)= \bigoplus \begin{cases}
0.5: \mathtt{pre}(2);  \mathtt{invoke}\  p(n-1);\\
0.5: \mathtt{pre}(1);  \mathtt{invoke}\  p(n-1);
\end{cases}
\end{equation*}
\qed
\end{example}

\begin{example}[MC1]\label{ex:rdwalk}
We consider the following manually-crafted recurrence relation,
\begin{equation*}
\mathtt{def}\ p(n;2)=  \{\mathtt{sample}\ v\leftarrow\mathtt{uniform}(n)\ \mathtt{in}\ \{\kwtick(1);\ \mathtt{invoke}\ p(v);\}\}
\end{equation*}
\qed
\end{example}

\begin{example}[MC2]\label{ex:rdwalk}
We consider the following manually-crafted recurrence relation,
\begin{equation*}
\mathtt{def}\ p(n;2)=  \{\mathtt{sample}\ v\leftarrow\mathtt{uniform}(n)\ \mathtt{in}\ \{\kwtick(\ln n);\ \mathtt{invoke}\ p(v);\}\}
\end{equation*}
\qed
\end{example}

\begin{example}[MC3]\label{ex:rdwalk}
We consider the following manually-crafted recurrence relation,
\begin{equation*}
\mathtt{def}\ p(n;2)=  \{\mathtt{sample}\ v\leftarrow\mathtt{muniform}(n)\ \mathtt{in}\ \{\kwtick(n\cdot \ln n);\ \mathtt{invoke}\ p(v);\}\}
\end{equation*}
\qed
\end{example}

\begin{example}[MC4]\label{ex:rdwalk}
We consider the following manually-crafted recurrence relation,
\begin{equation*}
\mathtt{def}\ p(n;2)= \bigoplus \begin{cases}
0.5: \{\mathtt{sample}\ v\leftarrow\mathtt{uniform}(n)\ \mathtt{in}\ \\\quad \quad \quad \quad \quad \{\kwtick(1);\ \mathtt{invoke}\ p(v);p(n-1-v);\}\};\\
0.5: \{\mathtt{sample}\ v\leftarrow\mathtt{muniform}(n)\ \mathtt{in}\ \{\kwtick(n);\ \mathtt{invoke}\ p(v);\}\}
\end{cases}
\end{equation*}
\qed
\end{example}

\section{The Detailed Semantics for PRRs}
\label{app:semantics}

Consider a PRR generated from the grammar in Figure~\ref{fig:grammar} with procedure name $p$, a \emph{configuration} $\sigma$ is a pair $\sigma=(\mbox{\sl comm}, \widehat{n})$ where $\mbox{\sl comm}$ is a statement that represents the current statement to be executed and $\widehat{n}\ge c_p$ is a positive integer that is the current value for the variable $n$.
A \emph{PRR state}  $\conf$ is a triple $\langle \sigma,\cost,\cont \rangle$ for which:
\begin{compactitem}
\item $\sigma$ is either the current configuration or $\mathsf{halt}$ that represents the termination of the whole PRR.
\item $\cost\ge 0$ records the cumulative preprocessing time so far.
\item $\cont$ is a stack of configurations that remain to be executed.
\end{compactitem}
 We use $\mathsf{emp}$ to denote an empty stack, and say that a PRR state $\langle \sigma,C,\cont \rangle$ is \emph{final} if $\cont=\emp$ and $\sigma=\mathsf{halt}$. Note that in a final configuration $\langle \mathsf{halt}, C,\emp \rangle$, the value $C$ represents the total execution runtime  of the PRR.

The semantics of the PRR is defined as a discrete-time Markov chain whose state space is the set of all PRR states and whose transition function $\transprob$ is defined in the way that the probability $\transprob(\conf, \conf')$ that the next PRR state is $\conf'$ given the current PRR state $\conf=((\mbox{\sl comm}, \widehat{n}),C,\cont)$ is determined by the following cases.
Below we abbreviate $e[\widehat{n}/n, \widehat{v}/v],e_i[\widehat{n}/n, \widehat{v}/v], (s-v)[\widehat{n}/n, \widehat{v}/v]$ as $\widehat{e},\widehat{e_i}, \widehat{s}-\widehat{v}$ respectively, and denote by $\mathsf{top}(\cont),\mathsf{pop}(\cont)$ the top element and resp. the stack after popping the top element in a stack $\cont$ respectively.

\begin{itemize}
\item In the case  $\mbox{\sl comm}=\mathsf{halt}$ and $\cont=\emp$, we define $\transprob(\conf,\conf):=1$ and $\transprob(\conf,\conf'):=0$ for other $\conf'$. This means that the PRR stays at termination once it terminates.
\item In the case $\mbox{\sl comm}=\kwsample\ v\leftarrow \mbox{\sl dist}\ \kwin\ \{\mathtt{pre}(e);\kwinvoke\ p(v);p(s-v)\}$ with distribution $\mbox{\sl dist}$, pseudo-polynomial expression $e$ and size expression $s$, we have that for every $\widehat{v}\in \supp(\mbox{\sl dist})$, $\transprob(\conf,\conf_{\widehat{v}}):=\mbox{\sl dist}(\widehat{v})$ for which
\[
\conf_{\widehat{v}}:=\begin{cases}
((\mathsf{func}(p),\widehat{v}),C+\widehat{e},(\mathsf{func}(p),\widehat{s}-\widehat{v})\cdot \cont) & \mbox{if }\widehat{v}, \widehat{s}-\widehat{v}\ge c_p\\
((\mathsf{func}(p),\widehat{v}),C+\widehat{e}, \cont) & \mbox{if }\widehat{v}\ge c_p~\&~\widehat{s}-\widehat{v}< c_p\\
((\mathsf{func}(p),\widehat{s}-\widehat{v}),C+\widehat{e}, \cont) & \mbox{if }\widehat{v}< c_p~\&~ \widehat{s}-\widehat{v}\ge c_p\\
(\mathsf{top}(\cont),C+\widehat{e}, \mathsf{pop}(\cont)) & \mbox{if }\widehat{v},\widehat{s}-\widehat{v}<c_p~\&~\cont\ne \emp \\
(\mathsf{halt},C+\widehat{e}, \emp) & \mbox{if }\widehat{v},\widehat{s}-\widehat{v}<c_p~\&~\cont=\emp \\
\end{cases}
\]
\item In the case $\mbox{\sl comm}=\kwsample\ v\leftarrow \mbox{\sl dist}\ \kwin\ \{\mathtt{pre}(e);\kwinvoke\ p(v)\}$ with distribution $\mbox{\sl dist}$ and pseudo-polynomial expression $e$. we have that for every $\widehat{v}\in \supp(\mbox{\sl dist})$, we define $\transprob(\conf,\conf_{\widehat{v}}):=\mbox{\sl dist}(\widehat{v})$ for which
\[
\conf'_{\widehat{v}}:=\begin{cases}
((\mathsf{func}(p),\widehat{v}),C+\widehat{e},\cont) & \mbox{if }\widehat{v}\ge c_p \\
(\mathsf{top}(\cont),C+\widehat{e},\mathsf{pop}(\cont)) & \mbox{if }\widehat{v}< c_p~\&~\cont\ne \emp\\
(\mathsf{halt}, C+\widehat{e},\emp) & \mbox{if }\widehat{v}< c_p~\&~\cont=\emp\\
\end{cases}
\]
The case of $\mbox{\sl comm}=\kwsample\ v\leftarrow \mbox{\sl dist}\ \kwin\ \{\mathtt{pre}(e);\kwinvoke\ p(n-v)\}$ can be obtained analogously.
\item  In the case $\mbox{\sl comm}=\kwwith_{i=1}^ke_i\!:\!\mbox{\sl comm}_i$,
we have that $\transprob(\conf, \conf_i) = \widehat{e_i}$ for each $1\le i\le k$ for which we have
$\conf_i:=((\mbox{\sl comm}_i, \widehat{n}),C,\cont)$.
\end{itemize}
With an initial PRR state
$((\mathsf{func}(p), n^*), 0, \emp)$
where $n^*\ge c_p$ is the input size, the Markov chain induces a probability space
where the sample space is the set of all infinite sequences of PRR states, the $\sigma$-algebra is generated by all sets of infinite sequences of PRR states that share a common prefix (called \emph{cylinder sets}), and the probability measure is uniquely determined by the probability transition function $\transprob$. We refer to \cite[Chapter~10]{BaierBook} for details. In the following, we denote the probability measure by $\Pr_{n^*}$ where $n^*\ge c_p$ is the input size.

\section{Details in Section \ref{sec:theory}}\label{appendix:theory}

\subsection{Details of the translation}

The definition of $\etf(\hat{n}, {\sl Prog})$ is
defined in Figure \ref{fig:etf}.
For recursive bodies, the translation is straightforward.
For sample, we first draw $\hat{v}$ from the distribution $\mathsf{dist}$ under the concrete input $\hat{n}$, and we substitute concrete values $\hat{n},\hat{v}$ with $n,v$ in the translation result for its recursive body.
We use a discrete distribution for probabilistic choice, where for each choice $1\le i\le k$, there is probability $\pp_i$ and returns the translation result for the sub-command $\com_i$. The following theorem formalizes its correctness.
\begin{theorem}
\label{thm:tf}
    Given a procedure $p(n;c_p)$, then for every $\hat{n}\ge c_p$, $\etf(\hat{n},\mathsf{func}(p))$ outputs the joint distribution of the triple of random variables $(\coststack(\hat{n}),\expr_1(\hat{n})$,$\\\expr_2(\hat{n}), r)$ in (\ref{eq:canonical-prr}).
\end{theorem}
\begin{proof}
The proof is straightforward by structural induction. The induction cases have been illustrated in the above text.
\end{proof}

\begin{figure}[htbp]
\begin{small}
\begin{talign*}
\textbf{(Command)  } &\etf(\hat{n}, \sample) = [\hat{n}/n,\hat{v}/v]\etf(\hat{n}, {\body})\\&\quad  where\ \hat{v}~is~a~meta~variable~drawn~from~the~distribution~[\hat{n}/n]\mathtt{dist}\\
 &\etf(\hat{n}, \with) \\  &\quad =  \mathsf{discrete}\left\{c_1:\etf(\hat{n}, \com_1),\dots,c_k:\etf(\hat{n}, \com_k)\right\}\\
\textbf{(Rec Body)  } &\etf(\hat{n}, \mathtt{pre}(\pp); \mathtt{invoke}\ p(\expr_1); p(\expr_2)) = ([\hat{n}/n]\pp,[\hat{n}/n]\expr_1,[\hat{n}/n]\expr_2, 2)\\
 &\etf(\hat{n}, \mathtt{pre}(\pp); \mathtt{invoke}\ p(\expr_1)) = ([\hat{n}/n]\pp,[\hat{n}/n]\expr_1,0, 1)\\
\end{talign*}
\end{small}
\caption{Transformation Function $\etf$}
\label{fig:etf}
\end{figure}

\subsection{Full Proof of Theorem \ref{thm:martingale}}

Below we present the full proof of Theorem \ref{thm:martingale}. To prove Theorem~\ref{thm:martingale}, we need the extension of conditional expectation to non-negative random variables that are not necessarily integrable. We adopt the method in the previous work~\cite{AgrawalC018}, where the conditional expectation is extended for non-negative random variables. A random variable $X$ is \emph{non-negative} if $X(\omega)\ge 0$ for all elements $\omega$ in the sample space.

\begin{theorem}[{\cite[Proposition~3.1]{AgrawalC018}}]\label{thm:condexpvext}
Let $X$ be any non-negative random variable from a probability space $(\Omega, \mathcal{F},\probm)$ and $\mathcal{G}$ be a sub-$\sigma$-algebra of $\mathcal{F}$ (i.e., $\mathcal{G}\subseteq\mathcal{F}$). Then there exists a random variable $\condexpv{X}{\mathcal{G}}$ from $(\Omega, \mathcal{F},\probm)$ (called a conditional expectation of $X$ w.r.t $\mathcal{G}$) that fulfills the two conditions (E1) and (E2) below:
\begin{itemize}
\item[(E1)] $\condexpv{X}{\mathcal{G}}$ is $\mathcal{G}$-measurable, and
\item[(E2)] for all $A\in\mathcal{G}$, we have $\int_A \condexpv{X}{\mathcal{G}}\,\mathrm{d}\!\probm=\int_A {X}\,\mathrm{d}\!\probm$.
\end{itemize}
Moreover,
the conditional expectation is almost surely unique, i.e., for any random variables $Y, Z$ from  $(\Omega, \mathcal{F},\probm)$ that both fulfill (E1) and (E2), we have that $\probm(Y=Z)=1$.
\end{theorem}
Below we show that some basic property we utilize in our proof is still valid upon the extension. Below we fix a probability space $(\Omega,\mathcal{F},\probm)$ and a sub-$\sigma$-algebra $\mathcal{G}\subseteq \mathcal{F}$.
The following properties hold for \emph{non-negative} random variables $X,Y,Z$.
\begin{itemize}
\item[(E3)] (`\textbf{Taking out what is known}') If $Z$ is $\mathcal{G}$-measurable, then $\condexpv{Z\cdot X}{\mathcal{G}}=Z\cdot\condexpv{X}{\mathcal{G}}$ a.s. The proof is exactly the same as under `Proof of (j)' on \cite[Page 90]{williams1991probability} for the non-negative case.
\end{itemize}

Our proof relies on the celebrated optional stopping theorem, which is stated below:
\begin{theorem}[Optional stopping theorem~{\cite[Theorem 10.10]{williams1991probability}}]
	\label{thm:optstp}
	For any discrete-time supermartingale $\dtsp=\{X_n\}_{n\in\Nset}$ adapted to a filtration
	$\{\mathcal{F}_n\}_{n\in\Nset}$ where every random variable $X_n$ is non-negative, and any stopping time $\rho$ w.r.t the filtration $\{\mathcal{F}_n\}_{n\in\Nset}$ that is almost surely finite (i.e.~$\probm(\rho<\infty)=1$), we have $\expv[X_\rho] \le \expv[X_0]$.
\end{theorem}

Now we are ready to prove Theorem \ref{thm:martingale}.

\begin{proof}[Proof of Theorem \ref{thm:martingale}]
We fix some sufficiently large $\alpha$ and $n^*$.
In general, we apply the martingale theory to prove this theorem. To construct a martingale, we need to make two preparations.

First, since $t_*(\alpha,n^*) \le t(\alpha,n)$ for every $c_p\le n\le n^*$, substituting $t(\alpha,n)$ with $t_*(\alpha,n^*)$ in (\ref{eq:constraint}) does not affect the validity of (\ref{eq:constraint}) by the convexity of $\exp(\cdot)$, i.e., $\forall c_p\le n\le n^*$, we have that:
\begin{talign}
    \label{eq:constraint2}
   &\mathbb E[\exp(t_*(\alpha,n^*)\cdot \ex(n\mid f))]\le t_*(\alpha,n^*)\cdot f(\alpha,n)
\end{talign}

Second, given an infinite sequence of the PRR states $\rho=\mu_0,\mu_1,\dots$ in the sample space, we construct its entry subsequence $\rho'=\mu'_0,\mu_1',\dots$ as follows. We keep all states $(({\sl comm},\hat{n}),C,\mathbf{K})\in \rho$ at the entry point of the procedure $p$, i.e., ${\sl comm} = \mathsf{func}(p)$, we also keep all final states, and remove all other states. We further define $\tau':=\inf\{t: \mu'_t\ \text{is final}\}$. Note that $C'_{\tau'}=C_{\tau}$. Furthermore, for every $i\ge 0$, if the state $\mu_i'$ is not a final state, then $\mu_{i+1}'$ represents the recursive calls of the state $\mu_i'$. Recall that by theorem \ref{thm:tf}, we have that the function $\etf(\hat{n},\mathsf{func}(p))$ outputs the joint distribution of the quadruple of random variables $(\coststack(\hat{n}),\expr_1(\hat{n}),\expr_2(\hat{n}),r)$ after executing $\mathsf{func}(p)$. Hence, if some state $\mu_i'=(({\sl comm}_i',\hat{n}_i'),C_i',\mathbf{K}_i')\in \rho'$ is not final, $\mu_{i+1}'$ observes as the following distribution.
\begin{itemize}
    \item Sample $(\Delta, \mathsf{s}_1, \mathsf{s}_2,r)$ from $\etf(\hat{n}_i',\mathsf{func}(p))$.
    \item If $r=1$,we have that:
    \begin{small}
    \begin{equation*}
        \mu_{i+1}' = \begin{cases}
        ((\mathsf{func}(p),\mathsf{s}_1),C_i'+\Delta, \mathbf{K}_i') & \mathsf{s}_1\ge c_p\\
        (\mathsf{top}(\mathbf{K}_i'),C_i'+\Delta,\mathsf{pop}(\mathbf{K}_i')) & \mathsf{s}_1< c_p\land \mathbf{K}_i'\ne \mathsf{emp}\\
        (\mathsf{halt},C_i'+\Delta, \mathsf{emp}) & \mathsf{s}_1< c_p\land \mathbf{K}_i'=\mathsf{emp}
        \end{cases}
    \end{equation*}
    \end{small}
    \item If $r=2$,we have that:
    \begin{small}
        \begin{equation*}
        \mu_{i+1}' = \begin{cases}
        ((\mathsf{func}(p),\mathsf{s}_1),C_i'+\Delta, (\mathsf{func}(p),\mathsf{s}_2)\cdot \mathbf{K}_i') & \mathsf{s}_1\ge c_p\land \mathsf{s}_2\ge c_p\\
        ((\mathsf{func}(p),\mathsf{s}_1),C_i'+\Delta, \mathbf{K}_i') & \mathsf{s}_1\ge c_p\land \mathsf{s}_2< c_p\\
        ((\mathsf{func}(p),\mathsf{s}_2),C_i'+\Delta, \mathbf{K}_i') & \mathsf{s}_1< c_p\land \mathsf{s}_2\ge c_p\\
        (\mathsf{top}(\mathbf{K}_i'),C_i'+\Delta,\mathsf{pop}(\mathbf{K}_i')) & \mathsf{s}_1< c_p\land \mathsf{s}_2<c_p\land \mathbf{K}_i'\ne \mathsf{emp}\\
        (\mathsf{halt},C_i'+\Delta, \mathsf{emp}) & \mathsf{s}_1< c_p\land \mathsf{s}_2<c_p\land \mathbf{K}_i'= \mathsf{emp}
        \end{cases}
    \end{equation*}
    \end{small}
\end{itemize}

Now we are ready to define the martingale. Given an entry subsequence $\rho'=\mu'_0,\mu_1',\dots$ in the sample space, for each $i\ge 0$, suppose $\mu_i'=(({\sl comm}_i',\hat{n}_i'),C_i',\mathbf{K}_i')$ and $\mathbf{K}_i' = (\mathsf{func}(p),\mathsf{s}_{i,1})\cdots  (\mathsf{func}(p),\mathsf{s}_{i,q_i})$, where $q_i$ represent the length of the stack. We define
\begin{talign*}
y_i:=\exp\left(C_i'+f(\alpha,\hat{n}_i')+\sum_{1\le j\le q_i} f(\alpha,\mathsf{s}_{i,j})\right)
\end{talign*}
Note that if $\mu_i'$ is not final, we have that:
\begin{talign*}
\mathbb E[y_{i+1}\mid y_i] &\le
y_i \cdot \mathbb E\left[\exp\left(\coststack(\hat{n}_i')+\sum_{j=1}^rf(\alpha,\expr_j(\hat{n}_i'))-f(\alpha,\hat{n}_i')\right)\right]\\
&\le y_i
\end{talign*}
where the first inequality is obtained by comparing $\mu_i'$ and $\mu_{i+1}'$ from the explicit formula of $\mu_{i+1}'$, and (E3) of the extended conditional expectation. The second inequality is from (\ref{eq:constraint2}).
Otherwise, if $\mu_i$ is final, then it is straightforward that $y_{i+1}=y_i=1$, thus we also have  $\mathbb E[y_{i+1}\mid y_i]\le y_i$ in this case.

Thus, $y_0,y_1,\ldots$ is a supermartingale, thus by the optional stopping theorem~(Theorem \ref{thm:optstp}), we have that
\begin{talign*}
\mathbb E\left[\exp\left(t_*(\alpha,n^*)\cdot C_{\tau}\right)\right]
&=\mathbb E\left[\exp\left(t_*(\alpha,n^*)\cdot C'_{\tau'}\right)\right]\\
&=\mathbb E\left[y_{\tau'} \right] \le \mathbb E[y_0] = \exp\left(t_*(\alpha,n^*)\cdot f(\alpha,n^*)\right)
\end{talign*}
\qed
\end{proof}

\subsection{Proof of Theorem \ref{thm:comp}}

\begin{proof}[Proof of Theorem \ref{thm:comp}]
The proof relies on Hoeffiding's Lemma as follows.
\begin{theorem}[Hoeffidng's Lemma,\cite{Hoeffding1963inequality}]
For every bounded random variable $a\le X\le b$, we have that $\mathbb E[\exp(X)]\le \exp(\mathbb E[X])\cdot \exp\left(\frac{(b-a)^2}{8}\right)$.
\end{theorem}

We first rephrase the constraint (\ref{eq:constraint}) as
\begin{talign*}
\mathbb E\left[\exp\left(t(\alpha,n)\cdot (\coststack(n)+\sum_{i=1}^rf(\alpha,\expr_i(n))-f(\alpha,n))\right)\right]\le 1
\end{talign*}

Define $X(\alpha,n)$ as the exponent in the $\exp(\cdot)$, $X(\alpha,n):=t(\alpha,n)\cdot (\coststack(n)+\sum_{i=1}^rf(\alpha,\expr_i(n))-f(\alpha,n))$, since we choose $f(\alpha,n)$ as $w(\alpha)\cdot \mathbb E[p(n)]$, and the fact that $\mathbb E[p(n)]=\mathbb E[\coststack(n)]+\sum_{i=1}^r \mathbb E[p(\expr_i)]$, we have that
\begin{talign*}
    \mathbb E[X(\alpha,n)] &= t(\alpha,n) \cdot (\mathbb E[\coststack(n)] + w(\alpha)\cdot \sum_{i=1}^r \mathbb E[p(\expr_i)] - w(\alpha)\cdot \mathbb E[p(n)])\\
    &= t(\alpha,n) \cdot (\mathbb E[\coststack(n)] - w(\alpha)\cdot \mathbb E[\coststack(n)])\\
    &= \frac{\lambda(\alpha)}{\mathbb E[\coststack(n)]} \cdot (1-w(\alpha))\cdot \mathbb E[\coststack(n)]\\
    &=-\lambda(\alpha)\cdot (w(\alpha)-1)
\end{talign*}
Furthermore, by (A1), we have that $X(\alpha,n)$ ranges from $\lambda(\alpha)\cdot (1-w(\alpha)+w(\alpha)\cdot M_1)$ to  $\lambda(\alpha)\cdot (1-w(\alpha)+w(\alpha)\cdot M_2)$. Hence, by the theorem above, we have that:
\begin{talign*}
\mathbb E[\exp(X(\alpha,n))]\le \exp\left(-\lambda(\alpha)\cdot (w(\alpha)-1) + \frac{\lambda(\alpha)^2\cdot w(\alpha)^2\cdot (M_2-M_1)^2}{8}\right)\le 1
\end{talign*}
where the second inequality is due to our choice of $\lambda(\alpha)$. Hence the theorem follows.
\qed

\end{proof}

\section{Details in Our Algorithmic Approach}
\label{appendix:check}
\subsection{Proof of Theorem~\ref{thm:binsearch}}
\begin{proof}
Suppose $c_t,c_f$ are valid and $c_t',c_f'$ are also valid, then due to the convexity, for any $0<\lambda<1$, $\lambda\cdot c_t +(1-\lambda)\cdot c_t',\lambda\cdot c_f +(1-\lambda)\cdot c_f'$ is also valid. Thus the second item holds since $0,c_f$ is always valid. The first items holds from the inequality $\sum_{i=1}^r f(\alpha,\expr_i(n))\le f(\alpha,n)$ if $\sum_{i=1}^r \expr_i\le n$.
\end{proof}

\subsection{The Algorithm~{\sl Decide}}

We first present some computer algebra backgrounds.

\paragraph{Monotonicity of pseudo-monomials} We need a complete depiction of the monotonicity of pseudo monomials, which could be straightforwardly proved by computing the derivative.

the monotonicity of a pseudo-monomial $n^a\cdot \ln^b n$ is presented in Table \ref{tab:mono}.

\begin{footnotesize}
\begin{table}[htbp]
\caption{Monotonicity of pseudo-monomials,
where we use $\uparrow\!(a,b)$ to represent the pseudo-monomial is non-decreasing when $n\in [a,b]$, use $\downarrow\!(a,b)$ to represent that the function is non-increasing when $n\in [a,b]$.}
\label{tab:mono}
\begin{center}
\begin{tabular}{|c|c|c|c|}
\hline
 & $b<0$ & $b=0$ & $b>0$\\
 \hline
$a<0$ & \multicolumn{2}{c|}{$\downarrow(c_p,\infty)$}  & \makecell[c]{$\uparrow(c_p,\exp\left(-\frac{b}{a}\right))$ and \smallskip  \\ $\downarrow(\exp\left(-\frac{b}{a}\right),\infty)$} \\
 \hline
$a=0$ & $\downarrow(0,\infty)$ & constant & \multirow{1}{*}{$\uparrow(c_p,\infty)$} \\
 \cline{1-4}
$a>0$ & \makecell[c]{$\downarrow(c_p,\exp\left(-\frac{b}{a}\right))$ and \smallskip  \\ $\uparrow(\exp\left(-\frac{b}{a}\right),\infty)$} & \multicolumn{2}{c|}{$\uparrow(c_p,\infty)$} \\
\hline
\end{tabular}
\end{center}
\end{table}
\end{footnotesize}

\paragraph{The function $\mathsf{NegativeLB}(P)$} Our decision procedure for canonical forms relies on a function $\mathsf{NegativeLB}(P)$, it takes on input as a pseudo-polynomial $P(n)$ whose sub-monomial term  with highest magnitude is negative, and outputs an integer $L^*$ such that $P(n)<0$ for every $n\ge L^*$. It is achieved by three steps below.
\begin{enumerate}
    \item Our algorithm firstly finds the term with largest order. In detail, our algorithm will scan each term in the polynomial and finds the term with highest degree of $n$, if there are several terms with the same highest degree of $n$, our algorithm chooses the term among them with the highest degree of $\ln n$. Formally, our algorithm finds the term $n^{a^*}\cdot \ln^{b^*} n$ with $c_{a^*, b^*}\ne 0$ and highest $a^*$, if there are more than one terms such that their degree of $n$ are all equal to $a^*$, we choose the largest $b^*$ among them. Then, our algorithm divides each term in $P(n)$ by $n^{a^*}\cdot \ln^{b^*} n$, and derives the pseudo-polynomial $P_1(n)$ below:
\begin{talign*}
P_1(n) = c_{a^*,b^*} + \sum_{a,b\ne a^*,b^*} c_{a,b}\cdot n^{a-a^*} \cdot \ln^{b-b^*} n
\end{talign*}
Note that by our choice of $a^*,b^*$, the degree of $n$ and $\ln n$ in each non-constant term in $P_1(n)$, i.e., $a-a^*,b-b^*$ either satisfies $a-a^*\le -1$, or satisfies $a-a^*=0$ and $b-b^*\le -1$. Furthermore, we have that $P(n)<0$ if and only of $P_1(n)<0$ when $n\ge c_p$. Hence, the transformation from $P(n)$ to $P_1(n)$ does not affect the answer. Below, we consider to compute $\mathsf{NegativeLB}(P_1)$.
\item Next, we consider an over-approximation $P_2(n)$ the function $P_1(n)$ as follows.
\begin{talign*}
P_2(n) := c_{a^*,b^*} + \sum_{a,b\ne a^*,b^*} \max\{0, c_{a,b}\}\cdot n^{a-a^*}\cdot \ln^{b-b^*} n
\end{talign*}
It removes all non-constant terms with negative coefficients. We prove that there exists a threshold $n_e$ such that $P_2(n)$ will monotonically decreasing for every $n\ge n_e$. Formally, we have that:
\begin{proposition}
    Let
    \begin{talign*}
    n_e:=\lceil \max\{\exp(-(b-b^*)/(a-a^*))\ |\ (a,b)\ne (a^*,b^*)\  \text{and}\  c_{a,b}\ne 0\}\rceil
    \end{talign*}
     be a constant. Then we have that $P_2(n)$ monotonically decreases for every $n\ge n_e$.
\end{proposition}
\begin{proof}
	Note that every coefficient for non-constant terms in $P_2(n)$ is non-negative. Thus the proposition directly follows from the monotonicity of the function $n^a\cdot \ln^b n$~(Table~\ref{tab:mono}).
\end{proof}
Since $P_1(n)\le P_2(n)$, we have that $\mathsf{NegativeLB}(P_2)$ also serves as an valid answer to $\mathsf{NegativeLB}(P_1)$. Blow, we consider to compute $\mathsf{NegativeLB}(P_2)$.
\item To compute $\mathsf{NegativeLB}(P_2)$, we simply enumerate every $n\ge n_e$ until we find some $n^*$ such that $P_2(n^*)< 0$. Since $P_2(n)$ will monotonically decreasing for every $n\ge n_e$, for every $n\ge n^*$, $P_1(n)\le P_2(n)\le P_2(n^*)<0$. Thus, $n^*$ is a valid answer to $\mathsf{NegativeLB}(P_1)$ and $\mathsf{NegativeLB}(P)$.
\end{enumerate}

Since sub-monomial term in $P(n)$ with highest magnitude has negative coefficient, we have that $c_{a^*,b^*}<0$. Hence, the function \textsf{NegativeLB}$(P)$ always terminates since as long as $n$ is sufficiently large, the $P_2(n)$ will be negative.

\paragraph{The decision procedure}
Our goal is to decide if $Q(\alpha,n)$ holds for every sufficiently large $\alpha,n^*$ and every $c_p\le n\le n^*$. Since $Q(\alpha,n)$ is irrelevant to $n^*$, it suffices to check whether the constraint holds over every sufficiently large $\alpha$ and every $n\ge c_p$. For each $1\le i\le k$, we use $S_i^{\alpha}$ to represent the sub-monomial term $f_i(\alpha)$ with highest magnitude, i.e., the sub-monomial term that firstly maximizes the degree of $\alpha$, and then maximizes the degree of $\ln \alpha$. Similarly, we use $S_i^{n}$ to represent the sub-monomial term $g_i(n)$ with highest magnitude.
The decision procedure consists of two steps. The pseudo-code is presented in Algorithm~\ref{alg:check}.

\begin{footnotesize}
\begin{algorithm}[htbp]\label{alg:check}
\SetInd{0.3em}{0.6em}
\caption{\small{The Decision procedure for canonical forms}}
    \label{alg:memoization}
    \SetKwFunction{FSearch}{\textsf{Decide}}
    \SetKwProg{Fn}{Procedure}{:}{}
    \SetKwInOut{Input}{Input}
    \SetKwInOut{Output}{Output}
    \Input{A constraint $Q(\alpha,n)$ in the canonical form~ (\ref{eq:canonical})}
    \Output{Decide whether $Q(\alpha,n)$ holds over every sufficiently large $\alpha$ and every $n\ge c_p$.}
    \Fn{\FSearch{$Q(\alpha,n)$}}{
        //The first step\\
        $L := c_p$\\
        \For{$i:=1,2,\ldots ,k$}{
             $S_i^{n}:=$ the sub-monomial term in $g_i(n)$ with highest magnitude. \\
            \If{$S_i^n>0$}{\textbf{Return} False}
            \Else{
             $g_i'(n):=$ the derivative of $g_i(n)$\\
            $L := \max\{L, {\sf NegativeLB}(g_i'(n))\}$
            }
        }
           //The second step\\
         \For{$\overline{n}:=c_p,\ldots ,L$}{
         $Q(\alpha,\overline{n}):=$ substitute $n$ with $\overline{n}$ in $Q(\alpha,n)$\\
         $R:=0$\\
         \For{$i:=1,2,\ldots, k$}{
             $S_i^{\alpha}:=$ the sub-monomial term in $f_i(\alpha)+g_i(\overline{n})$ with highest magnitude.\\
             Calculate $\Delta$ as the limit of $f_i(\alpha)+g_i(\overline{n})$ when $\alpha\to \infty$ by Proposition~\ref{prop:lim}.\\
             \If{$\Delta=\infty$}{\textbf{Return} False}
             \Else{$R:=R+\gamma_i\cdot \Delta$}
         }
         \lIf{$R\ge 1$}{\textbf{Return} False}
        }
        \textbf{Return} True
    }
\end{algorithm}
\end{footnotesize}

\smallskip \noindent $\bullet$ \emph{First}, in Lines 3--10, for each $1\le i\le k$, we
inspect the term $S_i^n$, if $S_i^n$ has positive coefficient, then clearly $f_i(n)\to \infty$ when $n\to \infty$ . Thus, the term $\exp(g_i(n)+f_i(\alpha))$ could be arbitrarily large. As a result, we can safely conclude that $Q(\alpha,n)$ does not hold (Lines 6--7). Otherwise, by computing $\mathsf{NegativeLB}(g_i'(n))$, we have that conclude that $\exp(g_i(n)+f_i(\alpha))$ will be non-increasing as long as $n\ge  \mathsf{NegativeLB}(g_i'(n))$. By taking the maximum over all $1\le i\le k$ (the maximum is stored in $L$) (Line 10), we have that the LHS of $Q(\alpha,n)$ will be non-increasing as long as $n\ge L$. Hence, it suffices to check every $c_p\le n\le L$.

\smallskip \noindent $\bullet$ \emph{Second}, for every integer $c_p\le \overline{n}\le L$, we substitute $n$ with $\overline{n}$ to eliminate $n$ in the exponent of the canonical form~(\ref{eq:canonical}). Then, each exponent $f_i(\alpha)+g_i(\overline{n})$ becomes a pseudo-polynomial solely over $\alpha$. Since we only concern sufficiently large $\alpha$, we can compute the limit $\Delta$ of $\exp(f_i(\alpha)+g_i(\overline{n}))$ when $\alpha\to \infty$ by the proposition (Proposition~\ref{prop:lim}) below, which could be proved straightforwardly.
\begin{proposition}\label{prop:lim}
The limit $\Delta$ of $\exp(f_i(\alpha)+g_i(\overline{n}))$ when $\alpha\to \infty$ is as follows.
\begin{compactitem}
    \item If $S_i^{\alpha}$ is superconsant and $S_i^{\alpha}>0$, then $\Delta=\infty$.
    \item If $S_i^{\alpha}$ is superconsant and $S_i^{\alpha}<0$, then $\Delta=0$.
    \item If $S_i^{\alpha}$ is constant $C$, then $\Delta=\exp(C)$.
    \item If $S_i^{\alpha}$ is subconstant,  then $\Delta=1$.
\end{compactitem}
\end{proposition}
By sum up the $\Delta$ for every $1\le i\le k$ as $R$, we compute the limit of the LHS of $Q(\alpha,\overline{n})$ when $\alpha\to \infty$.
We conclude that $Q(\alpha,n)$ if and only if $R<1$ for every $c_p\le \overline{n}\le L$. The property of our algorithm is as follows.
\begin{theorem}
The algorithm~(Algorithm~\ref{alg:check}) has the following property.
\begin{compactitem}
    \item If there exists two constants $\varepsilon,M>0$, such that for every $\alpha\ge M$ and every $n\ge c_p$, the LHS of $Q(\alpha,n)\le 1-\varepsilon$, then our algorithm will conclude that $Q(\alpha,n)$ holds.
    \item If $Q(\alpha,n)$ does not hold for sufficiently large $\alpha$ and every $n\ge c_p$, then our algorithm will conclude that $Q(\alpha,n)$ does not hold.
\end{compactitem}
\end{theorem}
\begin{proof}
Since $R$ is the limit of $Q(\alpha,n)$, the two items follows by the definition of limit. \qed
\end{proof}

\subsection{Details for the Divide-and-conquer Case}
Below we illustrate the divide-and-conquer case. We represent the branch as $\kwtick(S(n)); p(v); p(H_2(n)-v);$, where $H_2(n)$ is of the form $\lceil \frac nb\rceil +c$ or $\lfloor \frac nb\rfloor +c$. {\color{green} W.l.o.g., we assume that $H_2(n)$ is $n-1$ or $n$, and other choices of $H_2(n)$ follows the same treatment.} Similar to the single recursion case, we perform a case analysis on $\mathsf{dist}$.
When $\mathsf{dist}$ is an FSDPD, we follow the same strategy as (S1-D) and (S2-D). In detail, for (S1-D), we observe that $f(\alpha,v)+f(\alpha,H_2(n)-v)\le f(\alpha,H_2(n))\le f(\alpha,n)$, thus we can still over-approximate the expectation in~(\ref{eq:rephrase}) as $t(\alpha,n)\cdot S(n)$. For (S2-D), we only need to further consider to over-approximate the extra term $t(\alpha,n)\cdot f(\alpha, H_2(n)-v)$ by the same procedure.

When $\mathsf{dist}$ is $\mathtt{uniform}(n)$ (or $\mathtt{muniform}(n)$), then we can also over-approximate the subterm in the expectation~(\ref{eq:rephrase}) $L = \expv[\exp(t(\alpha,n)\cdot (f(\alpha,v)+f(\alpha,H_2(n)-v)))]$ through the summation as follows:
\begin{small}
\begin{talign}\label{eq:sum-divide}
    L \le \frac1n\sum_{i=0}^{n-1}\exp(t(\alpha,n)\cdot (f(\alpha,v)+f(\alpha,H_2(n)-v)))
\end{talign}
\end{small}
When $\mathsf{dist}$ is $\mathtt{muniform}(n)$, we simply replace $\frac1n$ by $\frac2n$. Since $f(\alpha,v)+f(\alpha,H_2(n)-v)$ is non-decreasing when $v\in [0,\lceil \frac{H_2(n)}{2}\rceil-1]$ and is non-increasing when $v\in [\lfloor \frac{H_2(n)}{2}\rfloor+1,H_2(n)]$. Hence, we can over-approximate the summation in~(\ref{eq:sum-divide}) as two integrals:
\begin{small}
\begin{talign*}
    &\left(\int_{0}^{\lceil \frac{H_2(n)}{2}\rceil}+\int_{\lfloor \frac{H_2(n)}{2}\rfloor}^{H_2(n)}\right)\exp(t(\alpha,n)\cdot (f(\alpha,v)+f(\alpha,H_2(n)-v))\text{d}v
    \\&=2\int_{0}^{\lceil \frac{H_2(n)}{2}\rceil}\exp(t(\alpha,n)\cdot (f(\alpha,v)+f(\alpha,H_2(n)-v))\text{d}v
\end{talign*}
\end{small}
The integral is more complex than the single recursion case, since we need to integrate the exponential function whose exponent is a pseudo-polynoimal over $n,\alpha,v$ and $H_2(n)-v$.
Thus, it is hard to apply our approach in the single recursion case.
Instead, we cut the integral into $Q$ parts (where $Q$ is also a parameter) with equal length and over-approximate each part using the maximum value.
Intuitively, the over-approximation will be more precise when $Q$ increases.
In this way, we can over-approximate the integral as:
\begin{talign*}
2\cdot \lceil \frac{H_2(n)}{2\cdot Q}\rceil\sum_{i=1}^{Q}\exp(t(\alpha,n)\cdot (f(\alpha,\frac{i}{Q}\cdot \lceil \frac{H_2(n)}{2}\rceil)+f(\alpha,H_2(n)-\frac{i}{Q}\cdot \lceil \frac{H_2(n)}{2}\rceil))
\end{talign*}
Note that the formula above is similar to the case of FSDPD, we apply the same strategy for FSDPDs to over-approximate the formula above into the canonical constraint.

\section{Details in the Empirical Evaluation~(Section~\ref{sec:evaluation})}
\label{app:evaluation}

\subsection{Detailed Synthesis Result}

In Table~\ref{table:detail}, we present the detailed synthesis results for each benchmark.

\renewcommand{\baselinestretch}{1.3}
\begin{table}[h]
\begin{center}
\begin{footnotesize}
    \caption{Detailed Synthesis Result}
    \label{table:detail}
    \begin{tabular}{|c|c|c|c|c|}
    \hline
        Benchmark & $\kappa(n^*)$ & Function & Tail bound $u(\alpha,n^*)$ & Time(s)\\
    \hline
        QuickSelect & $n^*$ & \makecell[c]{$\overline{f}(\alpha,n)=\frac{2\cdot \alpha}{\ln \alpha}\cdot n$\\$\overline{t}(\alpha,n)=\ln \alpha\cdot n^{-1}$} & $\exp(2\cdot \alpha-\alpha\cdot \ln\alpha)$ & $0.03$  \\
    \hline
        QuickSort & $n^*\!\cdot\! \ln n^*$ & \makecell[c]{$\overline{f}(\alpha,n)=4\cdot n\cdot \ln n$\\$\overline{t}(\alpha,n)=n^{-1}$} & $\exp((4-\alpha)\cdot \ln n)$ & $0.02$  \\
    \hline
        L1Diameter & $n^*$ & \makecell[c]{$\overline{f}(\alpha,n)=\frac{\alpha}{\ln \alpha}\cdot n$\\$\overline{t}(\alpha,n)=\ln \alpha\cdot n^{-1}$} & $\exp(\alpha-\alpha\cdot \ln\alpha)$ & $0.03$  \\
    \hline
        L2Diameter & $n^*\!\cdot\! \ln n^*$ & \makecell[c]{$\overline{f}(\alpha,n)=\frac{\alpha}{\ln \alpha}\cdot n\cdot \ln n$\\$\overline{t}(\alpha,n)=\ln \alpha\cdot n^{-1}\cdot \ln^{-1} n$} & $\exp(\alpha-\alpha\cdot \ln\alpha)$ & $0.03$  \\
    \hline
        RandSearch & $\ln n^*$ & \makecell[c]{$\overline{f}(\alpha,n)=\frac{2\cdot \alpha}{\ln \alpha}\cdot \ln n$\\$\overline{t}(\alpha,n)=\ln \alpha$} & $\exp((2\cdot \alpha-\alpha\cdot \ln\alpha)\cdot \ln n^*)$ & $0.03$  \\
    \hline
        Channel & $n^*$ & \makecell[c]{$\overline{f}(\alpha,n)=8\cdot n$\\$\overline{t}(\alpha,n)=1$} & $\exp(\frac12\cdot (8-\alpha)\cdot n^*)$ & $0.05$  \\
    \hline
        Rdwalk & $n^*$ & \makecell[c]{$\overline{f}(\alpha,n)=0.5\cdot n$\\$\overline{t}(\alpha,n)=1$} & $\exp((0.5-\alpha)\cdot n^*)$ & $0.05$  \\
    \hline
        Rdadder & $n^*$ & \makecell[c]{$\overline{f}(\alpha,n)=8\cdot n$\\$\overline{t}(\alpha,n)=\frac12$} & $\exp((4-0.5\cdot \alpha)\cdot n^*)$ & $0.04$  \\
    \hline
        MC1 & $\ln n^*$ & \makecell[c]{$\overline{f}(\alpha,n)=\frac{\alpha}{\ln\alpha}\cdot  \ln n$\\$\overline{t}(\alpha,n)=\ln \alpha$} & $\exp(\alpha-\alpha\cdot \ln \alpha)$ & $0.03$  \\
    \hline
        MC2 &$\ln^2 n^*$ &
        \makecell[c]{$\overline{f}(\alpha,n)=\frac{\alpha}{\ln\alpha}\cdot  \ln^2 n$\\$\overline{t}(\alpha,n)=\frac{\ln\alpha}{\ln n}$}
        &$\exp((\alpha - \alpha\ln \alpha)\cdot \ln n^*)$ & $0.03$  \\
    \hline
        MC3 &$n^*\cdot \ln^2 n^*$ &
        \makecell[c]{$\overline{f}(\alpha,n)=\frac{2\cdot \alpha}{\ln\alpha}\cdot  n\cdot \ln^2 n$\\$\overline{t}(\alpha,n)=\frac{\ln\alpha}{n\cdot \ln^2 n}$}
        &$\exp(2\cdot \alpha - \alpha\ln \alpha)$ & $0.03$  \\
    \hline
        MC4 & $n^*$ & \makecell[c]{$\overline{f}(\alpha,n)=\frac{2\cdot \alpha}{\ln \alpha}\cdot n$\\$\overline{t}(\alpha,n)=\ln\alpha\cdot n^{-1}$} & $\exp((2\cdot \alpha - \alpha\ln \alpha)\cdot \ln n^*)$ & $0.04$  \\
    \hline

\end{tabular}
\end{footnotesize}
\end{center}
\end{table}

\subsection{Concrete Bounds}

In Table~\ref{tab:concrete}, we present the concrete bounds. For each of our 12 benchmarks, we evaluate our synthesized
symbolic bound (Table~\ref{table:exp} in our paper) and Karp's symbolic bound with three concrete choices $(\alpha,n^*)=(10,13),(11,15),(12,17)$. We also evaluate the ratio between the two concrete bounds $\frac{\text{Karp's Bound}}{\text{Our Bound}}$ to show that our approach can derive much tighter bounds.

\begin{table}[!ht]
    \centering
    \caption{Results for Concrete Bounds}
    \label{tab:concrete}
    \begin{tabular}{|l|l|l|l|l|}
    \hline
        Benchmark & Concrete choice & Our bound & Karp's Bound & Ratio \\ \hline
        QuickSelect & $\alpha=10;n^*=13$ & $4.85*10^{-2}$ & $1.92*10^{-1}$ & $3.96$ \\ \hline
        QuickSelect & $\alpha=11;n^*=15$ & $1.26*10^{-2}$ & $1.45*10^{-1}$ & $1.16*10^{1}$ \\ \hline
        QuickSelect & $\alpha=12;n^*=17$ & $2.97*10^{-3}$ & $1.10*10^{-1}$ & $3.69*10^{1}$ \\ \hline
        QuickSort & $\alpha=10;n^*=13$ & $2.07*10^{-7}$ & $1.11*10^{-2}$ & $5.36*10^{4}$ \\ \hline
        QuickSort & $\alpha=11;n^*=15$ & $5.85*10^{-9}$ & $6.74*10^{-3}$ & $1.15*10^{6}$ \\ \hline
        QuickSort & $\alpha=12;n^*=17$ & $1.43*10^{-10}$ & $4.09*10^{-3}$ & $2.85*10^{7}$ \\ \hline
        L1Diameter & $\alpha=10;n^*=13$ & $2.20*10^{-6}$ & $4.05*10^{-3}$ & $1.84*10^{3}$ \\ \hline
        L1Diameter & $\alpha=11;n^*=15$ & $2.10*10^{-7}$ & $2.03*10^{-3}$ & $9.67*10^{3}$ \\ \hline
        L1Diameter & $\alpha=12;n^*=17$ & $1.83*10^{-8}$ & $1.02*10^{-3}$ & $5.58*10^{4}$ \\ \hline
        L2Diameter & $\alpha=10;n^*=13$ & $2.20*10^{-6}$ & $4.05*10^{-3}$ & $1.84*10^{3}$ \\ \hline
        L2Diameter & $\alpha=11;n^*=15$ & $2.10*10^{-7}$ & $2.03*10^{-3}$ & $9.67*10^{3}$ \\ \hline
        L2Diameter & $\alpha=12;n^*=17$ & $1.83*10^{-8}$ & $1.02*10^{-3}$ & $5.58*10^{4}$ \\ \hline
        RandSearch & $\alpha=10;n^*=13$ & $4.26*10^{-4}$ & $5.88*10^{-4}$ & $1.38$ \\ \hline
        RandSearch & $\alpha=11;n^*=15$ & $7.12*10^{-6}$ & $1.77*10^{-4}$ & $2.49*10^{1}$ \\ \hline
        RandSearch & $\alpha=12;n^*=17$ & $6.92*10^{-8}$ & $5.22*10^{-5}$ & $7.55*10^{2}$ \\ \hline
        Channel & $\alpha=10;n^*=13$ & $2.61*10^{-23}$ & $6.72*10^{-2}$ & $2.57*10^{21}$ \\ \hline
        Channel & $\alpha=11;n^*=15$ & $4.84*10^{-30}$ & $4.64*10^{-2}$ & $9.59*10^{27}$ \\ \hline
        Channel & $\alpha=12;n^*=17$ & $1.22*10^{-37}$ & $3.21*10^{-2}$ & $2.64*10^{35}$ \\ \hline
        Rdwalk & $\alpha=10;n^*=13$ & $1.98*10^{-13}$ & $3.02*10^{-2}$ & $1.52*10^{11}$ \\ \hline
        Rdwalk & $\alpha=11;n^*=15$ & $5.18*10^{-17}$ & $2.00*10^{-2}$ & $3.87*10^{14}$ \\ \hline
        Rdwalk & $\alpha=12;n^*=17$ & $4.97*10^{-21}$ & $1.33*10^{-2}$ & $2.67*10^{18}$ \\ \hline
        Rdadder & $\alpha=10;n^*=13$ & $1.98*10^{-13}$ & Not applicable & - \\ \hline
        Rdadder & $\alpha=11;n^*=15$ & $5.18*10^{-17}$ & Not applicable & - \\ \hline
        Rdadder & $\alpha=12;n^*=17$ & $4.97*10^{-21}$ & Not applicable & - \\ \hline
        MC1 & $\alpha=10;n^*=13$ & $3.09*10^{-15}$ & $2.06*10^{-8}$ & $6.67*10^{6}$ \\ \hline
        MC1 & $\alpha=11;n^*=15$ & $8.23*10^{-19}$ & $1.18*10^{-9}$ & $1.44*10^{9}$ \\ \hline
        MC1 & $\alpha=12;n^*=17$ & $1.19*10^{-22}$ & $6.48*10^{-11}$ & $5.46*10^{11}$ \\ \hline
        MC2 & $\alpha=10;n^*=13$ & $3.09*10^{-15}$ & $2.06*10^{-8}$ & $6.67*10^{6}$ \\ \hline
        MC2 & $\alpha=11;n^*=15$ & $8.23*10^{-19}$ & $1.18*10^{-9}$ & $1.44*10^{9}$ \\ \hline
        MC2 & $\alpha=12;n^*=17$ & $1.19*10^{-22}$ & $6.48*10^{-11}$ & $5.46*10^{11}$ \\ \hline
        MC3 & $\alpha=10;n^*=13$ & $4.85*10^{-2}$ & $1.92*10^{-1}$ & $3.96$ \\ \hline
        MC3 & $\alpha=11;n^*=15$ & $1.26*10^{-2}$ & $1.45*10^{-1}$ & $1.16*10^{1}$ \\ \hline
        MC3 & $\alpha=12;n^*=17$ & $2.97*10^{-3}$ & $1.10*10^{-1}$ & $3.69*10^{1}$ \\ \hline
        MC4 & $\alpha=10;n^*=13$ & $2.20*10^{-6}$ & Not applicable & - \\ \hline
        MC4 & $\alpha=11;n^*=15$ & $2.10*10^{-7}$ & Not applicable & - \\ \hline
        MC4 & $\alpha=12;n^*=17$ & $1.83*10^{-8}$ & Not applicable & - \\ \hline
    \end{tabular}
\end{table}

\subsection{Plots for Concrete Bounds}

In Figures~\ref{fig:b1}--\ref{fig:b12}, we present the plots of our concrete bounds and the Karp's concrete bounds over $10\le \alpha\le 15,n^*=17$.
\begin{figure}
	\centering
	\includegraphics[width=0.7\textwidth]{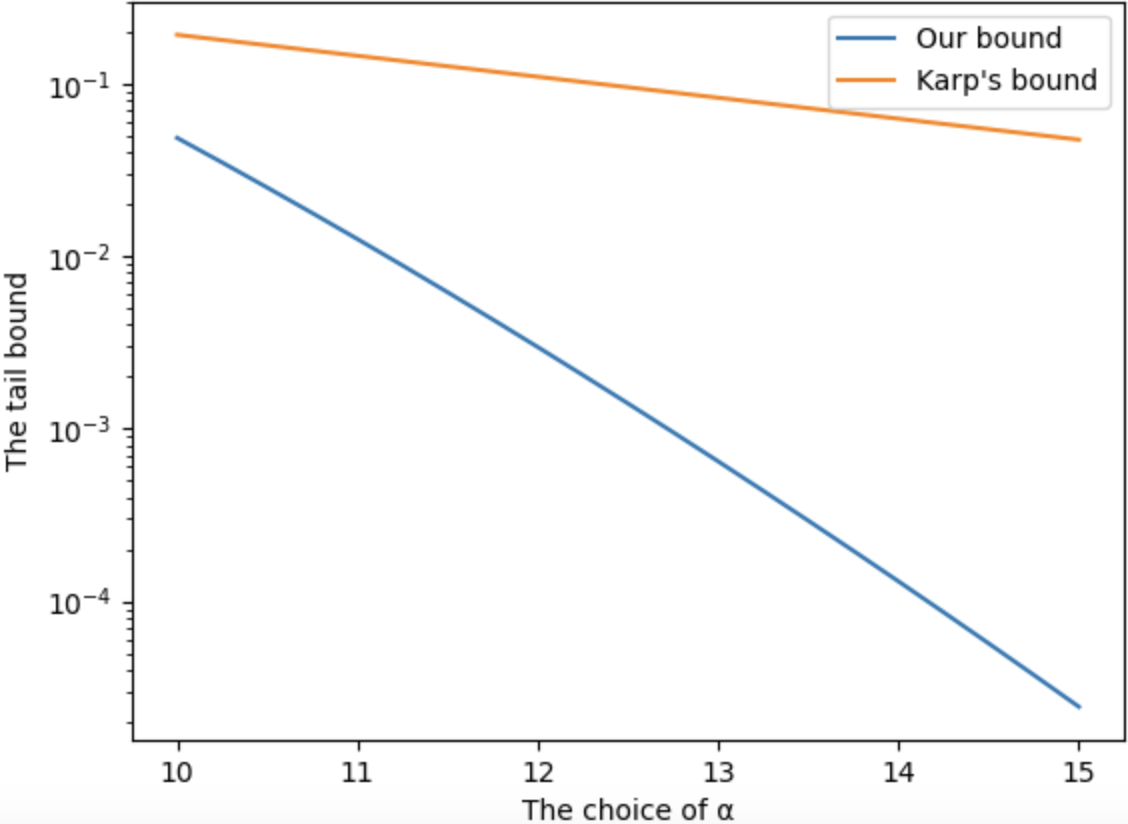}
	\caption{The Plot for QuickSelect}
	\label{fig:b1}
\end{figure}
\begin{figure}
	\centering
	\includegraphics[width=0.7\textwidth]{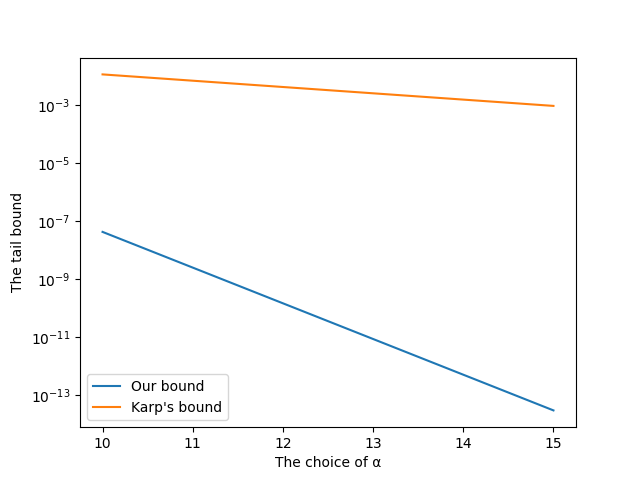}
	\caption{The Plot for QuickSort}
	\label{fig:b2}
\end{figure}
\begin{figure}
	\centering
	\includegraphics[width=0.7\textwidth]{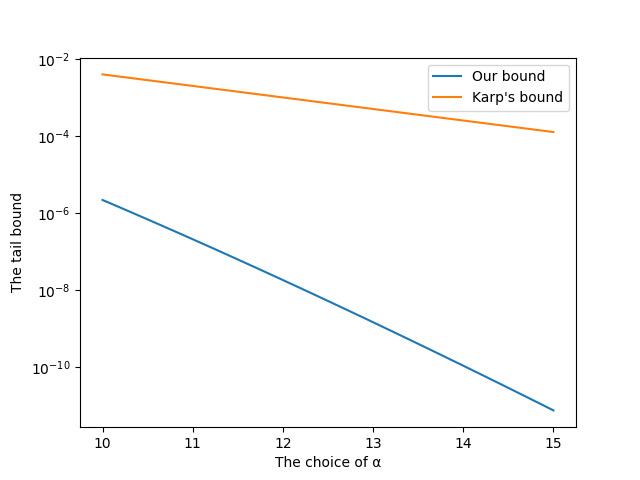}
	\caption{The Plot for L1Diameter}
	\label{fig:b3}
\end{figure}
\begin{figure}
	\centering
	\includegraphics[width=0.7\textwidth]{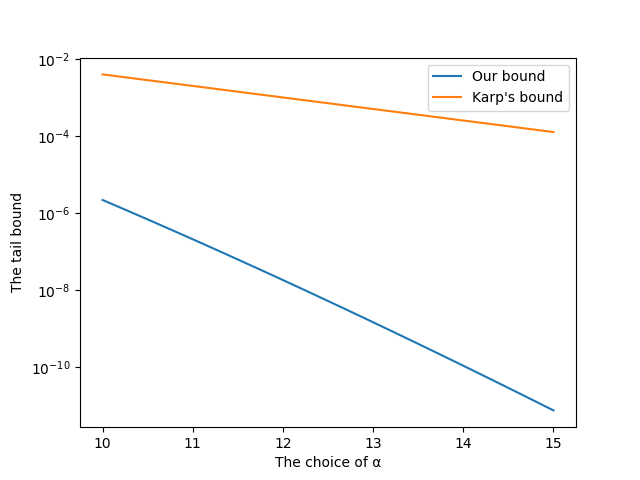}
	\caption{The Plot for L2Diameter}
	\label{fig:b4}
\end{figure}
\begin{figure}
	\centering
	\includegraphics[width=0.7\textwidth]{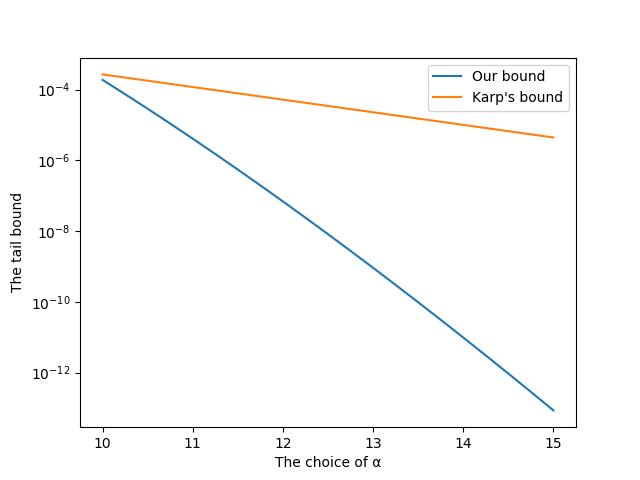}
	\caption{The Plot for RandSearch}
	\label{fig:b5}
\end{figure}
\begin{figure}
	\centering
	\includegraphics[width=0.7\textwidth]{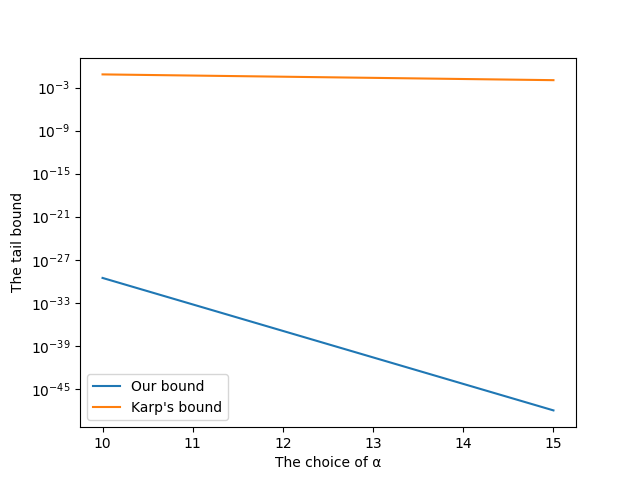}
	\caption{The Plot for Channel}
	\label{fig:b6}
\end{figure}
\begin{figure}
	\centering
	\includegraphics[width=0.7\textwidth]{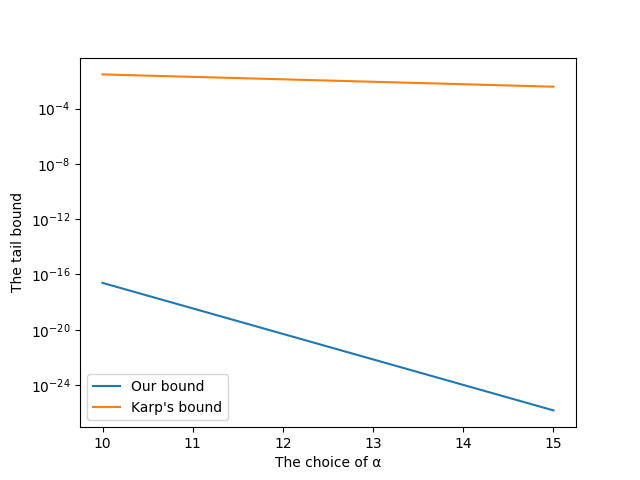}
	\caption{The Plot for Rdwalk}
	\label{fig:b7}
\end{figure}
\begin{figure}
	\centering
	\includegraphics[width=0.7\textwidth]{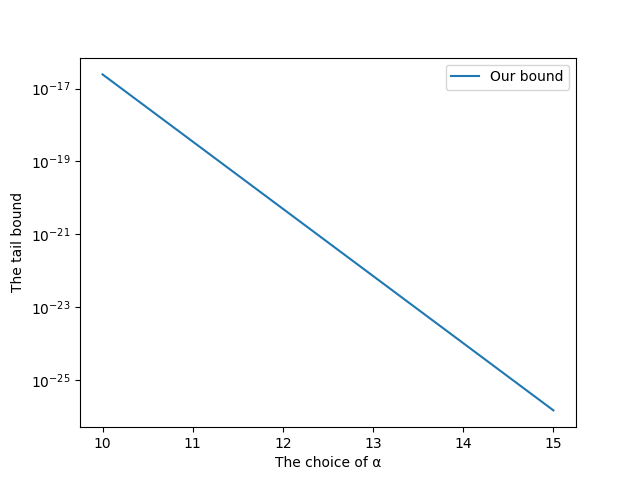}
	\caption{The Plot for Rdadder}
	\label{fig:b8}
\end{figure}
\begin{figure}
	\centering
	\includegraphics[width=0.7\textwidth]{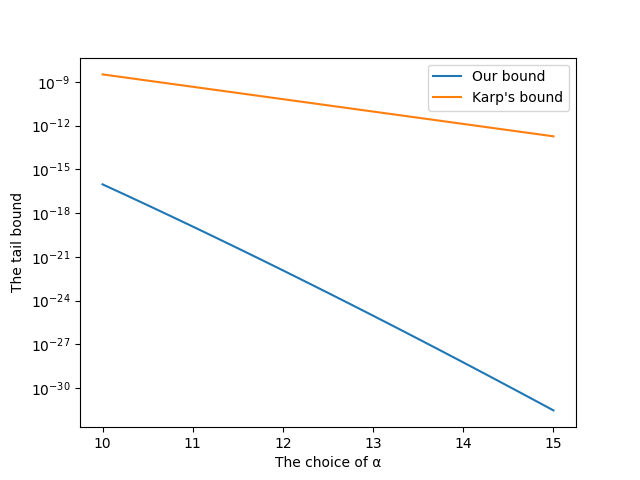}
	\caption{The Plot for MC1}
	\label{fig:b9}
\end{figure}
\begin{figure}
	\centering
	\includegraphics[width=0.7\textwidth]{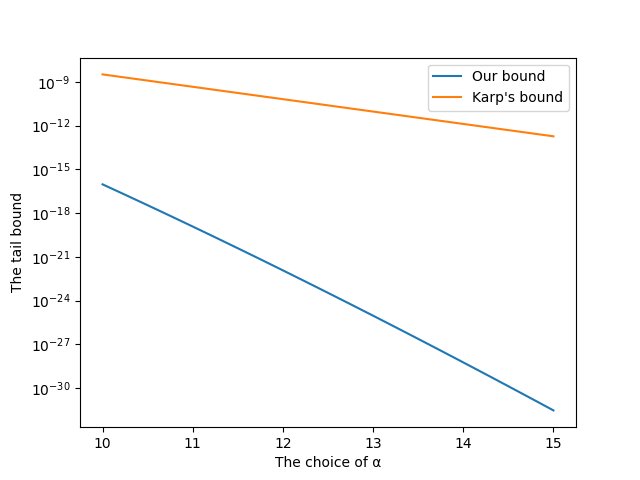}
	\caption{The Plot for MC2}
	\label{fig:b10}
\end{figure}
\begin{figure}
	\centering
	\includegraphics[width=0.7\textwidth]{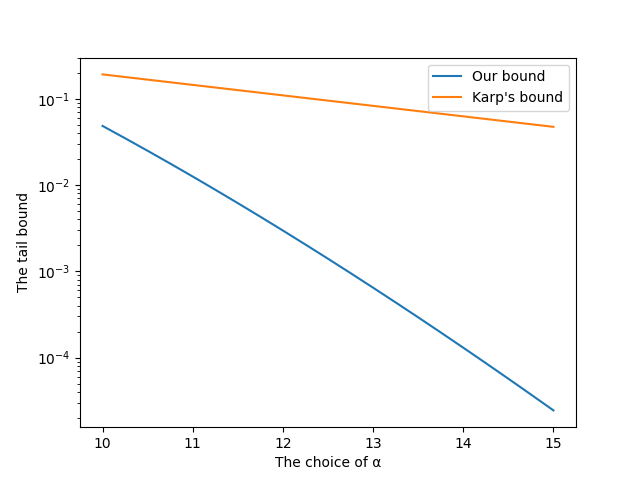}
	\caption{The Plot for MC3}
	\label{fig:b11}
\end{figure}
\begin{figure}
	\centering
	\includegraphics[width=0.7\textwidth]{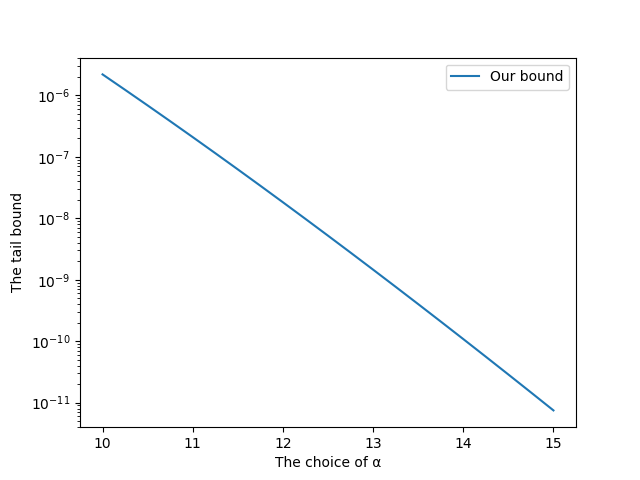}
	\caption{The Plot for MC4}
	\label{fig:b12}
\end{figure}

\subsection{Tail Bounds by Our Completeness Theorem~\ref{thm:comp}}

In Table~\ref{table:compthm}, we present the synthesized tail bound from our completeness theorem~(Theorem~\ref{thm:comp}).

\renewcommand{\baselinestretch}{2.0}
\begin{table}[h]
\begin{center}
\begin{footnotesize}
    \caption{Tail Bound By the Completeness Theorem}
    \label{table:compthm}
    \begin{tabular}{|c|c|c|c|c|}
    \hline
        Benchmark & $\kappa(n^*)$ & Tail bound $u(\alpha,n^*)$\\
    \hline
        QuickSelect & $4n^*$  & $\exp(-2\cdot \frac{(\alpha-1)^2}{\alpha})$  \\
    \hline
        QuickSort & $2n^*\!\cdot\! \ln n^*$ & $\exp(-0.7\cdot \frac{(\alpha-1)^2}{\alpha}\cdot \ln n^*)$ \\
    \hline
        L1Diameter & $2n^*$ & $\exp(- \frac{(\alpha-1)^2}{\alpha})$  \\
    \hline
        L2Diameter & $2n^*\!\cdot\! \ln n^*$ & $\exp(-\frac{(\alpha-1)^2}{\alpha})$  \\
    \hline
        RandSearch & $\ln n^*/\ln \frac43$ & $\exp(-1.19\cdot \frac{(\alpha-1)^2}{\alpha}\cdot \ln n^*)$  \\
    \hline
        Channel & $en^*$ & $\exp(-0.74\cdot \frac{(\alpha-1)^2}{\alpha}\cdot n^*)$  \\
    \hline
        Rdwalk & $0.5n^*$ & $\exp(-\frac13\cdot \frac{(\alpha-1)^2}{\alpha}\cdot n^*)$\\
    \hline
        Rdadder & $3n^*$& $\exp(-\frac{(\alpha-1)^2}{\alpha}\cdot n^*)$  \\
    \hline
        MC1 & $\ln n^*/\ln 2$&Not applicable  \\
    \hline
        MC2 & $\ln^2 n^*/\ln 2$&Not applicable  \\
    \hline
        MC3 & $2n^*\cdot \ln^2 n^*$ & $\exp(-\frac{(\alpha-1)^2}{\alpha})$  \\
    \hline
        MC4 & $2.5n^*$ & $\exp(-\frac{(\alpha-1)^2}{\alpha})$ \\
    \hline

\end{tabular}
\end{footnotesize}
\end{center}
\end{table}

\section{Missing Details in Related Work~(Section~\ref{sec:related})}
\label{app:notfixedpoint}
\subsection{Coupon Collector}

The example above shows that our approach can obtain bounds that are considerably tighter than Karp's method.
Moreover, we will now extend the results to systems of PRRs consisting of several recursive equations (modeling recursive algorithms with several procedures). This is also beyond the reach of Karp's method and all previous methods. However, our approach is unable to handle a program, called {CouponCollector}\cite[Chapter 3.6]{DBLP:books/cu/MotwaniR95} as follows.
\begin{example}
Consider the Coupon-Collector problem with $n$ different types of coupons$(n \in \mathbb N)$.
The randomized process proceeds in rounds: at each round, a coupon is collected uniformly at random from the coupon types the rounds continue until all the $n$ types of coupons are collected.
We model the rounds as a recurrence relation with two variables $n,m$, where $n$ represents the total number of coupon types and $m$ represents the remaining number of uncollected coupon types.
The recurrence relation is as follows:
\begin{equation*}
p(n,m;1) = \bigoplus \begin{cases}
m/n: \mathtt{pre}(1); \mathtt{invoke}\  p(n,m-1);\\
1-m/n: \mathtt{pre}(1); \mathtt{invoke}\  p(n,m);
\end{cases}
\end{equation*}
\end{example}
Our approach fails in case since it is hard to apply Markov's inequality to derive exponentially decreasing tail bounds. However, Karp's method could produce the following:
$$
\Pr[p(n,n)\ge n\ln n + \alpha\cdot n]\le \exp(-\alpha).
$$
\qed

\subsection{Non-Prefixed-Point of Karp's Bound on Quick-Select}
\label{app:rw}

Consider the PRR for Quick-Select:
$$\textstyle
\mathtt{def}\ p(n;2)= \{\mathtt{sample}\ v\leftarrow\mathtt{muniform}(n)\ \mathtt{in}\ \{\kwtick(n);\ \mathtt{invoke}\ p(v);\}\}
$$

where $h(n)$ is the muniform distribution.
Karp's method gives the following tail bound:

$$
\Pr(p(n)> 4\cdot n + k\cdot n)\le \left(\frac{3}{4}\right)^k.
$$

Although Karp's result is derived (indirectly) from fixed point theory, we show that the derived tail bound $\left(\frac{3}{4}\right)^k$ is not a prefixed point when $k$ is allowed to be any non-negative real number.
First, recall that the prefixed point condition for the situation $k\ge 2$ and an even number $n$ states that

$$
\left(\frac{3}{4}\right)^k\ge 2\cdot \frac{1}{n} \sum_{i=\frac{n}{2}+1}^n \left(\frac{3}{4}\right)^{(3+k) -4\cdot \frac{i}{n}}.
$$

But we have
\begin{eqnarray*}
2\cdot \frac{1}{n}\cdot \sum_{i=\frac{n}{2}}^{n-1}\left(\frac{3}{4}\right)^{(3+k) -4\cdot \frac{i}{n}} &=& 2\cdot \frac{1}{n}\cdot \left(\frac{3}{4}\right)^{k+1}\cdot \frac{1-[(\frac{3}{4})^{-\frac{4}{n}}]^{\frac{n}{2}}}{1-(\frac{3}{4})^{-\frac{4}{n}}}\\
&=& 2\cdot \frac{1}{n}\cdot\left(\frac{3}{4}\right)^{k+1}\cdot \color{red}{\frac{-\frac79}{1-e^{\frac{4}{n}\cdot \ln \frac43}}}\\
&\approx&2\cdot \frac{1}{n}\cdot \left(\frac{3}{4}\right)^{k+1}\cdot \frac{\frac79}{\frac{4}{n}
\cdot \ln \frac43}\\
&\approx&2\cdot \frac{1}{n}\cdot n\cdot \left(\frac{3}{4}\right)^{k+1} \cdot {\color{red} 0.676}\\
&=&(2\cdot 0.75\cdot 0.676)\cdot \left(\frac34\right)^k \ge 1.01\cdot \left(\frac34\right)^k
\end{eqnarray*}
for which the first $\approx$ are with the ratio $1$ when $n$ tends to infinity. 
Hence, we have
$$
\left(\frac{3}{4}\right)^k< 2\cdot \frac{1}{n} \sum_{i=\frac{n}{2}+1}^n \left(\frac{3}{4}\right)^{(3+k) -4\cdot \frac{i}{n}}.
$$
As a result, the prefixed point condition is violated for every sufficiently large even number $n$ and $k\ge 2$.

\end{document}